\documentclass[a4paper,11pt]{article}

\usepackage{fullpage}

\usepackage{amsmath,amssymb}
\usepackage{graphicx}
\usepackage{color}
\usepackage{algorithm}
\usepackage{algorithmic}
\usepackage{url}
\usepackage{wasysym}
\usepackage{subcaption}
\usepackage{hyperref}
\usepackage{authblk}
\usepackage{mathtools}
\usepackage[shortlabels]{enumitem}
\usepackage{wrapfig}

\graphicspath{{./Figures/}}

\newtheorem{theorem}{Theorem}

\newtheorem{corollary}{Corollary}
\newtheorem{lemma}{Lemma}
\newtheorem{remark*}{Remark}

\newcommand{\qedclaim}{\hfill $\diamond$ \medskip}

\newenvironment{proof}{\par \noindent {\bf Proof}.\ }{\hfill$\blacksquare$
\par \vspace{11pt}}

\pdfoutput=1 %uncomment to ensure pdflatex processing (mandatatory e.g. to submit to arXiv)
\graphicspath{{./Figures/}}%helpful if your graphic files are in another directory

\renewcommand{\Pr}{\mathbb{P}}
\newcommand{\E}{{\mathbb{E}}}
\newcommand{\V}{{\mathbb{V}ar}}

\usepackage{thm-restate}

\newtheorem{repeat_lemma_3}{Lemma 2}
\newtheorem{repeat_lemma_4}{Lemma 3}
\newtheorem{repeat_lemma_5}{Lemma 4}
\newtheorem{repeat_theorem_6}{Theorem 2}
\newtheorem{repeat_lemma_8}{Lemma 7}
\newtheorem{repeat_lemma_11}{Lemma 8}
\newtheorem{repeat_lemma_12}{Lemma 9}
\newtheorem{repeat_lemma_13}{Lemma 10}

\author[1]{Fr{\'e}d{\'e}ric Giroire}
\author[1]{Nicolas Nisse}
\author[1]{Kostiantyn Ohulchanskyi}
\author[1,2]{Ma{\l}gorzata Sulkowska}
\author[1]{Thibaud Trolliet}
\affil[1]{Universit{\'e} C{\^o}te d’Azur, CNRS, Inria, I3S, France}
\affil[2]{Wroc{\l}aw University of Science and Technology, Department~of~Fundamentals~of~Computer Science, Poland}

\date{}

\title{Preferential attachment hypergraph with vertex deactivation}

\begin{document}
\maketitle 

\begin{abstract}
In the field of complex networks, hypergraph models have so far received significantly less attention than graphs. However, many real-life networks feature multiary relations (co-authorship, protein reactions) may therefore be modeled way better by hypergraphs. Also, a recent study by Broido and Clauset suggests that a power-law degree distribution is not as ubiquitous in the natural systems as it was thought so far. They experimentally confirm that a majority of networks (56\% of around 1000 %social, biological, technological, transportation, and information
networks that undergone the test) favor a power-law with an exponential cutoff over other distributions. We address the two above observations by introducing a preferential attachment hypergraph model which allows for vertex deactivations. The phenomenon of vertex deactivations is rare in existing theoretical models and omnipresent in real-life scenarios (social network accounts which are not maintained forever, collaboration networks in which people retire, technological networks in which devices break down). We prove that the degree distribution of the proposed model follows a power-law with an exponential cutoff. We also check experimentally that a Scopus collaboration network has the same characteristic. We believe that our model will predict well the behavior of systems from a variety of domains.
\end{abstract}

%%%%%%%%%%%%%%%%%%%%%%%%%%%%%%%%%%%%%%%%
\section{Introduction}

The notion of \emph{complex networks} relates to the mathematical structures modeling large real-life systems. %It is impressive how wide stretches their application 
Their omnipresence across different life domains is remarkable. Complex networks model biological networks (e.g., protein or gene interactions schemes, maps of neural connections in the brain), social networks (Facebook, Twitter, Snapchat, collaboration networks), technological networks (power grids, transportation networks), the World Wide Web, etc. %Nowadays complex networks became ubiquitous as they not only model,
They allow to predict the behavior of the systems, serve as the benchmarks for testing algorithms that are used later in the real networks, and, in general, allow to understand better the underlying mechanisms that create those systems in nature. Roughly, since 1999, one observes a dynamical growth in experimental and theoretical research on complex networks in computer science, mathematical, and physical societies. It was the year when Barab{\'a}si and Albert introduced the seminal model of a \emph{preferential attachment} random graph~\cite{BA_basic}. This model is based on two mechanisms: growth (the graph is growing over time, gaining a new vertex and a bunch of edges at each time step) and preferential attachment (an arriving vertex is more likely to attach to other vertices with high degrees rather than with low degrees). It captures the \emph{small world} (small diameter) and the \emph{rich get richer} (leading to a heavy tailed degree distribution) phenomena commonly observed in nature. 

Since then, a number of theoretical models were presented, e.g., \cite{Watts_Strogatz,Molloy_Reed,Chung_Lu_model,Cooper_Frieze,Buck_Ost}. These were mostly graph models concentrated on reflecting three phenomena: a small diameter, a high clustering coefficient, and a \emph{power-law} degree distribution. It was thought for a long time that a power-law degree distribution is the most commonly present in nature~\cite{Bol_Rio_chapter}. However, this statement was recently questioned by Broido and Clauset \cite{BrCl2019}. They performed statistical tests on almost 1000 social, biological, technological, transportation, and information networks and observed that a ``majority of networks (56\%) favor the power-law with cutoff model over other distributions''. The cutoff observed in the tail of a distribution may be caused by a finite-size character of the dynamic network, i.e., when the elements deactivate after some time \cite{BrCl2019}. %Here one has to mention that in most of the classical models introduced so far nodes remain active forever which is also not true in many real-life scenarios.
The phenomenon of vertex deactivations is rare in known theoretical models and omnipresent in real-life scenarios as the extinction events are fundamental in the world surrounding us. Think of social networks (Facebook, Twitter, Instagram, etc.) where users unsubscribe or simply stop using them, collaboration networks in which deactivated nodes represent people who retired, died or stopped working in the given domain, technological networks where a vertex deactivation is interpreted as a breakdown of the device or the web network in which web pages are not maintained forever.
%Models that allow for vertex deactivation seem more natural than those in which elements remain active forever. 
%In fact, deactivation of vertices is more natural than keeping them alife forever .
Even though some theoretical models featuring deletions or deactivations of vertices were introduced \cite{CoFrVe2004,MoGhNe2006}, just a few of them lead to a degree distribution following a power-law with an exponential cutoff. One of the widely cited is a balls and bins scheme %with ball deactivation
introduced by Fenner et al.~\cite{FeLeLo2005,FeLeLo2007}. %may serve as an example where  with ball deactivation is studied. %However, since it is a balls and bins model, it keeps

In the model from \cite{FeLeLo2005,FeLeLo2007}, information about the degree of each element of the network is kept but information about who is connected with whom is lost. Working with graphs instead of bins and balls allows to keep this information. Nevertheless, graphs have another clear limitation. They reflect only binary relations while in practice we encounter many higher order relations (groups of interest, protein reactions, co-authorship, interactions between biological cells, GitHub users committing to the same repository). Nowadays they are often modeled in graphs by cliques which may lead to a profound information loss \cite{beyond_pairwise_survey_2020}. E.g., if there are three researchers in a triangle in a collaboration graph, one cannot tell whether they published one paper together or three independent papers, each per pair of researchers. Higher order relations can be captured by hypergraphs, that is, a generalization of graphs in which each (hyper)edge possibly links together more than two nodes. Sometimes keeping information about hyperedge may have a profound impact on analyzing the model and drawing conclusions. Consider any example in which a big hyperedge strongly indicates belonging to the same community (e.g., an email sent to a group of people should evidence the existence of a community rather than be treated as a set of bilateral emails). So far hypergraph models have received significantly less attention than graphs in the area of complex networks. %First, they were introduced under the name random intersection graphs \cite{rand_inters_models,rand_inters_properties,beyond_pairwise_survey_2020}.
Wang et al. introduced a preferential attachment hypergraph model 
%A preferential attachment hypergraph model was first introduced by Wang et al. introduce a prefer \cite{Wang_hyper}. However, %the family considered there it was
but restricted to a specific subfamily of uniform acyclic hypergraphs (the analogue of trees within graphs) \cite{Wang_hyper}. The first rigorously studied non-uniform hypergraph preferential attachment model was proposed only in 2019 by Avin et al. \cite{ALP_hyper} and featured a power-law degree distribution. Another dynamic hypergraph model with a clear community structure was presented in \cite{hyper_mod}. Note that there exists an analogy between the hypergraphs and the random intersection graphs \cite{rand_inters_models,rand_inters_properties}. The algorithms and software tools for working with hypergraph networks, even the definitions of some features and measures started appearing only recently \cite{hyper_jl,hypernetx,KaPoPrSzTh2019,KaPrTh2020}.

\textbf{Results.} We propose a preferential attachment hypergraph model in which vertices may become inactive after some time.  %(e.g., a neuron may die in a brain, a researcher from the collaboration network may retire).
The hyperedges model multiary and not necessarily uniform relations, we allow for different cardinalities of hyperedges (e.g., articles may have different numbers of co-authors). %We prove that the degree distribution of our model follows a power-law with an exponential cutoff and give another real-life example, a Scopus research collaboration network, with this distribution.
We prove that the degree distribution of our model follows a power-law with an exponential cutoff and compare it with a real-life example, a Scopus research collaboration network. We believe that our model will be the next step towards developing the hypergraph chapter in the complex networks area and that will serve as a useful tool predicting well the behavior of the systems from a variety of domains.

\textbf{Paper organization.} Section \ref{sec:notation} contains basic definitions and notation. In Section~\ref{sec:model}, we introduce the hypergraph model with vertex deactivation and prove that its degree distribution follows a power-law with an exponential cutoff using a master equation approach. Due to the presence of hyperedges and to the possibility of vertex deactivation, we had to modify the classical approach (e.g. treat active and inactive vertices separately) and take advantage of some tools that were not used in this context before (e.g. the Stolz-Cesaro Theorem). In Section \ref{sec:theta}, we estimate one of the parameters that appears in the formula for the degree distribution of our model. It is defined as a limit, existence of which we assume (Assumption~(4) formulated in Section~3.2) to prove the main result (Theorem \ref{thm:deg_dist}). %\fred{either put "defined as a limit" or "defined as the limit of the average (over time) expected degree of vertices chosen for deactivation"}
%The technical novelty is a treatment of assumption~(4) formulated in Section 3.2 needed to prove the main result (Theorem \ref{thm:deg_dist}). 
Such an assumption was already present in the literature on models with degree distribution following a power-law with an exponential cutoff, \cite{FeLeLo2005,FeLeLo2007}. Even though we also did not manage to prove the existence of this limit directly, the technical novelty is that we give a formally rigid indication on how to estimate its value (using Gaussian hypergeometric functions and the Banach Fixed Point Theorem) and an experimental justification for its existence in Section~\ref{sec:experiments}. Section~\ref{sec:experiments} also includes the experimental results on real data and the simulations of the model. Further works are discussed in Section~\ref{sec:conclude}. %Technical proofs are presented in the Appendix.

%Networks with exponential cut-off: \\
%> yeast protein interaction network \cite{FeLeLo2005} \\
%> Mathematical Research collaboration network \cite{FeLeLo2007} \\
%> our scientific collaboration network Section \ref{sec:experiments}\\
%> sublinear attachment rule gives rise to power-law with stretched exponential (more subtle than just exponential cutoff) \cite{Krapivsky_2000,Krapivsky_2001}

\section{Basic definitions and notation} \label{sec:notation}

We define a \emph{hypergraph} $H$ as a pair $H=(V,E)$, where $V$ is a set of vertices and $E$ is a multiset of hyperedges, i.e., non-empty, unordered multisets of $V$. We allow for a multiple appearance of a vertex in a hyperedge (self-loops) as well as a multiple appearance of a hyperedge in $E$. The degree of a vertex $v$ in a hyperedge $e$, denoted by $d(v,e)$, is the number of times $v$ appears in $e$. The cardinality of a hyperedge $e$ is  $|e| = \sum_{v \in e} d(v,e)$. The degree of a vertex $v \in V$ in $H$ is understood as
the number of times it appears in all hyperedges, i.e., $\deg(v) = \sum_{e \in E} d(v,e)$. If $|e|=k$ for all $e \in E$, $H$ is said to be {\it $k$-uniform}.

We consider hypergraphs that grow by adding vertices and/or hyperedges at discrete time steps $t=0,1,2,\ldots$ according to some rules involving randomness. The random hypergraph obtained at time $t$ will be denoted by $H_t=(V_t,E_t)$ and the degree of $u \in V_t$ in $H_t$ by $\deg_t(u)$. During this building process some of the  vertices may become deactivated.  %(deactivated vertex stays deactivated forever). 
Therefore the set $V_t$ splits into $\mathcal{A}_t$, the set of vertices active at time $t$ (denote its cardinality by $A_t$), and $\mathcal{I}_t$, the set of vertices that are not active at time $t$ (denote its cardinality by $I_t$); thus $|V_t| = A_t + I_t$. By $D_t$ we denote the sum of degrees of vertices active at time~$t$, i.e., $D_t = \sum_{u \in \mathcal{A}_t} \deg_t(u)$. Moreover, we write $\Theta_t$ for the degree of a vertex chosen for deactivation at time $t$ (the description of a deactivation procedure is given within the formal definition of the model in the next section).

$N_{k,t}$ stands for the number of vertices in $H_t$ of degree $k$. Thus $\sum_{k \geq 1} N_{k,t} = |V_t|$. Similarly, $A_{k,t}$ is the number of active vertices of degree $k$ at time $t$ and $I_{k,t}$ the number of inactive vertices of degree $k$ at time $t$ (denote the corresponding sets by $\mathcal{A}_{k,t}$ and $\mathcal{I}_{k,t}$, respectively); $\sum_{k \geq 1} A_{k,t} = A_t$, $\sum_{k \geq 1} I_{k,t} = I_t$ and $N_{k,t} = A_{k,t} + I_{k,t}$. We write $f(k) \sim g(k)$ if $f(k)/g(k) \xrightarrow{k \rightarrow \infty} 1$. 
%For $f$ and $g$ being real functions we write $f(k) \sim g(k)$ if $f(k)/g(k) \xrightarrow{k \rightarrow \infty} 1$.
We say that the degree distribution of a random hypergraph follows a \emph{power-law} if the expected fraction of vertices of degree $k$ is proportional to $k^{-\beta}$ for some exponent $\beta > 1$. Formally, we interpret it as $\lim_{t \rightarrow \infty} \E\left[\frac{N_{k,t}}{|V_t|}\right] \sim c \cdot k^{-\beta}$ for some positive constants $c$ and $\beta > 1$. Similarly, we say that the degree distribution of $H_t$ follows a \emph{power-law with an exponential cutoff} if $\lim_{t \rightarrow \infty} \E\left[\frac{N_{k,t}}{|V_t|}\right] \sim c \cdot k^{-\beta} \gamma^k$, where $\gamma \in (0,1)$.

We say that an event $A$ occurs \emph{with high probability} (whp) if the probability $\Pr[A]$ depends on a certain number $t$ and tends to $1$ as $t$ tends to infinity.

\section{Preferential attachment hypergraph with vertex deactivation} \label{sec:model}

The model introduced in this section may be seen as a generalization of a hypergraph model presented by Avin et al. in \cite{ALP_hyper}. The model from \cite{ALP_hyper} allows for two different actions at a single time step - attaching a new vertex by a hyperedge to the existing structure or creating a new hyperedge on already existing vertices. We add another possibility - deactivation of a vertex. Once a vertex is chosen for deactivation, it stays deactivated forever, i.e., it remains in the hypergraph but it can not be chosen to the new hyperedges - its degree freezes and the hyperedges incident with it remain in the hypergraph. Avin et al. proved that the degree distribution of their model follows a power-law. We prove that adding the possibility of deactivation of vertices generates an exponential cutoff in the degree distribution.

\subsection{Model $\mathbf{H(H_0,p_v,p_e,Y)}$}
The hypergraph model $H$ is characterized by the following  parameters:  
\begin{enumerate}
	\item $H_0$ - the initial hypergraph, seen at $t=0$;
	\item $p_v, p_e, p_d = 1-p_e-p_v$ - the probabilities indicating, what are the chances that a particular type of event occurs at a single time step;
	\item $Y = (Y_0, Y_1, \ldots, Y_t, \ldots)$ - independent random variables, giving the cardinalities of the hyperedges that are added at a single time step.
\end{enumerate}
Here is how the structure of $H = H(H_0,p_v,p_e,Y)$ is being built. We start with some non-empty hypergraph $H_0$ at $t=0$. We assume for simplicity that $H_0$ consists of a hyperedge of cardinality $1$ over a single vertex. Nevertheless, all the proofs may be generalized to any initial $H_0$ having constant number of vertices and constant number of hyperedges with constant cardinalities. `Vertices chosen from $\mathcal{A}_t$ in proportion to degrees' means that active vertices are chosen independently (possibly with repetitions) and the probability that any $u$ from $\mathcal{A}_t$ is chosen is
\[
\Pr[u \textnormal{ is chosen}] = \frac{\deg_t(u)}{\sum_{v \in \mathcal{A}_t}\deg_t(v)} = \frac{\deg_t(u)}{D_t}
\]
($\deg_t(u)$ and $\deg_t(v)$ refer to the degrees of $u$ and $v$ in the whole $H_t$). For $t \geqslant 0$ we form $H_{t+1}$ from $H_t$ choosing only one of the following events according to $p_v, p_e, p_d$.

\begin{itemize}
	\item With probability $p_v$: Add one vertex $v$. Draw a value $y$ being a realization of $Y_t$. Then select $y-1$ vertices from $\mathcal{A}_t$ in proportion to degrees; add a new hyperedge consisting of $v$ and the $y-1$ selected vertices.
	\item With probability $p_e$: Draw a value $y$ being a realization of $Y_t$. Then select $y$ vertices from $\mathcal{A}_t$ in proportion to degrees; add a new hyperedge consisting of the $y$ selected vertices.
	\item With probability $p_d$: Choose one vertex from $\mathcal{A}_t$ in proportion to degrees. Deactivate it, i.e., $\mathcal{A}_{t+1} = \mathcal{A}_t \setminus  \{v\}$ and $\mathcal{I}_{t+1} = \mathcal{I}_t \cup  \{v\}$.
\end{itemize}

\begin{remark*}
	Note that this model can be simplified to many known models by choosing the appropriate set of parameters:
	\begin{enumerate}[1)]
		\item setting $p_v=1$, $p_e=p_d=0$ and $Y_t = 2$ (all the hyperedges are of size $2$ thus one simply builds a graph) one gets the Barab{\'a}si-Albert tree \cite{BA_basic};
		\item setting $p_d=0$ and $Y_t=2$ one gets the preferential attachment scheme for graphs with vertex- and edge-step~\cite{ChLu_book}, Chapter~3;
		\item setting $p_d=0$ one gets the hypergraph model presented by Avin et al. in~\cite{ALP_hyper}.
	\end{enumerate}
	%Note that setting some parameters trivially to constants or zeroes boils down the whole setting to the known models. E.g., setting   E.g., setting $p_d=0$ one gets the presented in $\cite{ALP_hyper}$.
\end{remark*}

%\begin{remark*}
%	After setting $p_d=0$ above, the model boils down to the one presented in $\cite{ALP_hyper}$.
%\end{remark*}

\begin{remark*}
	As the hypergraph gets large, the probability of creating a self-loop can be well bounded and is quite small provided that the sizes of hyperedges are reasonably bounded.
\end{remark*}

Note that if we want a process to continue then it is reasonable to demand that, on average, we add more vertices to the system than we deactivate. Therefore we always assume $p_v > p_d$. Then the probability that the process %terminates (i.e., that we arrive at the moment in which all vertices are deactivated) equals $(p_d/p_v)^i$
will not terminate (i.e., that we never arrive at the moment in which all vertices are deactivated) is positive and equals $1-(p_d/p_v)^i$, where $i$ is the number of active vertices at time $t=0$, in our case $i=1$ (compare with the probability that the gambler's fortune will increase forever, \cite{Epstein_book}). %Throughout the whole paper
We concentrate only on the case when the process does not terminate. 
%(compare with the probability that the gambler's fortune won't increase forever, \cite{Epstein_book}). %Throughout the paper we assume that we 
%Whenever it happens, we restart the simulation.

\subsection{Degree distribution of $\mathbf{H(H_0,p_v,p_e,Y)}$}

In this section we prove that the degree distribution of $H = H(H_0,p_v,p_e,Y)$ follows a power-law with an exponential cutoff under four assumptions.

First two of them address the distributions of the cardinalities of hyperedges ($Y_t$) added step by step. We assume that their expectation is constant and their variance sublinear in $t$, which, we feel, is in accordance with many real-life systems (in particular, with the scientific collaboration network we are working with experimentally in Section \ref{sec:experiments}).  %that they are reasonably small in comparison with the sum of degrees in the whole hypergraph. 

The third assumption tells that we will restrict ourselves to only such distributions of $Y_t$ for which the distribution of $D_t$ (the sum of degrees of active vertices at time $t$) remains concentrated. Similar assumption one finds in other papers on complex network models, e.g. in  \cite{ALP_hyper} by Avin et al. (presenting a model of a preferential attachment hypergraph with the degree distribution following a power-law) or in \cite{Krapivsky_2001,Krapivsky_2000} by Krapivsky et al. (where the models in which the arriving vertex attaches to the existing node $w$ with probability proportional to $(\deg{w})^r$ with $r<1$ is studied).

The fourth assumption refers to the average sum of degrees of vertices chosen for deactivation ($\sum_{\tau=1}^{t} \E[\Theta_{\tau}]$). In Section \ref{sec:theta} we prove that its order is $\Theta(t)$. However, we additionally assume that  the limit $\lim_{t \rightarrow \infty} \frac{1}{t} \sum_{\tau=1}^{t} \E[\Theta_{\tau}]$ exists and equals some $\theta \in \mathbb{R}_{>0}$. Such assumption was also already present in the literature on models with degree distribution following a power-law with an exponential cutoff \cite{FeLeLo2005,FeLeLo2007}. Since we were not able to (just as the authors of \cite{FeLeLo2005} or \cite{FeLeLo2007}) theoretically justify the existence of the stated limit we support it by simulations in Section \ref{sec:experiments}. We also explain in Section \ref{sec:theta} how the limiting value may be obtained, assuming that the limit exists.

%The fourth assumption was also already present in the literature on models with degree distribution following a power-law with an exponential cutoff (consult \cite{FeLeLo2005,FeLeLo2007}). It assumes the existence of the limit $\lim_{t \rightarrow \infty} \frac{1}{t} \sum_{\tau=1}^{t} \E[\Theta_{\tau}]$. Throughout the paper we prove that the average sum of degrees of vertices chosen for deactivation ($\sum_{\tau=1}^{t} \E[\Theta_{\tau}]$) is of order $\Theta(t)$, however we are not able to (just as the authors of \cite{FeLeLo2005} or \cite{FeLeLo2007}) theoretically justify the existence of the stated limit. We leave it as the assumption strongly supported by simulations in Section \ref{sec:experiments}. We also explain in Section \ref{sec:theta} how the limiting value may be obtained, assuming that the limit exists.

\noindent
\textbf{Assumptions}
\begin{enumerate} 
	\item $\E[Y_t] = \mu \in \mathbb{R}_{>0}$ for all $t>0$.
	\item $\V[Y_t] = o(t)$.
	\item $\Pr[D_t \neq \E[D_t] + o(t)] = o(1/t)$.
	%\item $D_t = \E[D_t] + o(t)$ whp.
	%\item $\E\left[\frac{Y_t^2}{D_{t-1}^2}\right] = o\left(\frac{1}{t}\right)$.
	\item $\lim_{t \rightarrow \infty} \frac{1}{t} \sum_{\tau=1}^{t} \E[\Theta_{\tau}] = \theta \in \mathbb{R}_{>0}$.
\end{enumerate}

%From now on by $A_{k,t}$ we denote the number of active vertices of degree $k$ in $H_t$ and, analogously, by $I_{k,t}$ the number of deactivated vertices of degree $k$ in $H_t$. Recall that $N_{k,t}$ stands for the number of vertices of degree $k$ in $H_t$. Thus $N_{k,t} = A_{k,t} + I_{k,t}$.

Before we formally state and prove the main theorem we introduce several technical lemmas and theorems that will be helpful later on.

\begin{theorem}[Stolz-Ces{\`a}ro theorem] \label{thm:S-C}
	Let $(a_t)_{t \geq 1}$ and $(b_t)_{t \geq 1}$ be the sequences of real numbers. Assume that $(b_t)_{t \geq 1}$ is strictly monotone and divergent. If  $\lim_{t \rightarrow \infty} \frac{a_{t+1}-a_t}{b_{t+1}-b_t} = g$ then $\lim_{t \rightarrow \infty} \frac{a_t}{b_t} = g$.
%	\[
%	\lim_{t \rightarrow \infty} \frac{a_{t+1}-a_t}{b_{t+1}-b_t} = g \quad \quad then \quad \quad \lim_{t \rightarrow \infty} \frac{a_t}{b_t} = g.
%	\]
	%then %the limit $\lim_{t \rightarrow \infty} \frac{a_t}{b_t}$ exists and 
%	\[
%	\lim_{t \rightarrow \infty} \frac{a_t}{b_t} = g.%\lim_{t \rightarrow \infty} \frac{a_{t+1}-a_t}{b_{t+1}-b_t} = g.
%	\]
\end{theorem}

%\begin{lemma}[Chernoff Bounds, \cite{MU_book}] \label{lemma:Chernoff}
%	Let $X = \sum_{i=1}^t X_i$, where $X_i = 1$ with probability $p$ and $X_i = 0$ with probability $1-p$, and all $X_i$ are independent. Let $m = \E[X] = t p$. Then
%	\[
%	\Pr[|X-m| \geq \delta m] \leq 2 e^{-m\delta^2/3}
%	\]
%	for all $\delta \in (0,1)$.
%\end{lemma}

\begin{lemma}[\cite{ChLu_book}, Chapter 3.3] \label{lemma:rec_seq}
	Let $(a_t)_{t \geq 1}$, $(b_t)_{t \geq 1}$ and $(c_t)_{t \geq 1}$ be the sequences of real numbers, where $b_t \xrightarrow{t \rightarrow \infty} b>0$, $c_t \xrightarrow{t \rightarrow \infty} c$ and $a_t$ satisfies the recursive relation
	%Let $\{a_t\}$ be a sequence satisfying the recursive relation
	$a_{t+1}~=~\left(1-\frac{b_t}{t}\right)a_t + c_t$.
	%where $b_t \xrightarrow{t \rightarrow \infty} b>0$ and $c_t \xrightarrow{t \rightarrow \infty} c$.
	Then $\lim_{t \rightarrow \infty} \frac{a_t}{t} = \frac{c}{1+b}$.
	%$\lim_{t \rightarrow \infty} \frac{a_t}{t} = \frac{c}{1+b}$. %the limit $\lim_{t \rightarrow \infty} \frac{a_t}{t}$ exists and
%	\[
%	\lim_{t \rightarrow \infty} \frac{a_t}{t} = \frac{c}{1+b}.
%	\]
\end{lemma}

The proofs of Lemmas \ref{lemma:Nkt_limit}, \ref{lemma:Y2/D2}, and \ref{lemma:A/D} can be found in the Appendix A.

\begin{lemma} %[See Lemma 4 in \cite{hyper_mod}]
	\label{lemma:Nkt_limit}
	If $\lim_{t \rightarrow \infty} \frac{\E[N_{k,t}]}{t} \sim c \cdot k^{-\beta} \gamma^k \left(\frac{1}{k} + \delta \right)$ for some positive constants $c, \beta, \gamma, \delta$ then $\lim_{t \rightarrow \infty} \E\left[\frac{N_{k,t}}{|V_t|}\right] \sim \frac{c}{p_v} k^{-\beta} \gamma^k \left(\frac{1}{k} + \delta \right)$.
%	\[
%	\lim_{t \rightarrow \infty} \E\left[\frac{N_{k,t}}{|V_t|}\right] \sim \frac{c}{p_v} k^{-\beta} \gamma^k \left(\frac{1}{k} + \delta \right).
%	\]
	\textnormal{(}Here ``$\sim$'' refers to the limit by $k \rightarrow \infty$.\textnormal{)}
\end{lemma}
%The proof can be found in the Appendix.

\begin{lemma} \label{lemma:Y2/D2}
	Assume that $\E[Y_t] = \mu$ for all $t>0$ and $\V[Y_t] = o(t)$. Then
	\[
	\E\left[\frac{Y_t^2}{D_{t-1}^2}\right] = o\left(\frac{1}{t}\right).
	\]
	 %and $D_t = \E[D_t] + o(t)$ whp. Then 
%	\[
%	\E\left[\frac{Y_t^2}{D_{t-1}^2}\right] = o\left(\frac{1}{t}\right).
%	\]
\end{lemma}

\begin{lemma} \label{lemma:A/D}
	Assume that $D_t = \E[D_t] + o(t)$ whp. Then for each $k \geq 1$
	\[
	\E\left[\frac{A_{k,t}}{D_t}\right] = \frac{\E[A_{k,t}]}{\E[D_t]} + o(1).
	\] 
%	\[
%	\E\left[\frac{A_{k,t}}{D_t}\right] = \frac{\E[A_{k,t}]}{\E[D_t]} + o(1).
%	\]
\end{lemma}

\begin{lemma} \label{fact:Dt/t}
	Let $\E[Y_t] = \mu$ for all $t>0$. %and let $\Theta_t$ denote the degree of a vertex chosen for deactivation at time $t$.
	Assume that  $\lim_{t \rightarrow \infty} \frac{1}{t} \sum_{\tau=1}^{t} \E[\Theta_{\tau}] = \theta \in \mathbb{R}_{>0}$. Then $\lim_{t \rightarrow \infty} \frac{\E[D_t]}{t} = (p_v+p_e)\mu - p_d \theta$.
%	\[
%		\lim_{t \rightarrow \infty} \frac{\E[D_t]}{t} = (p_v+p_e)\mu - p_d \theta.
%	\]
\end{lemma}
\begin{proof}
	The initial hypergraph $H_0$ consists of a single vertex of degree $1$. Since at time $t \geq 1$ we add a hyperedge of cardinality $Y_t$ with probability $p_v+p_e$ and we deactivate a vertex of degree $\Theta_t$ with probability $p_d$ we get
	\begin{equation}  \label{eq:Dt_Theta}
		\E[D_t] = 1 + (p_v+p_e)\sum_{\tau=1}^{t}\E[Y_{\tau}] - p_d \sum_{\tau=1}^{t}\E[\Theta_{\tau}].
	\end{equation}
	The conclusion follows.
\end{proof}

\begin{theorem} \label{thm:deg_dist}
	Consider a hypergraph $H = H(H_0,p_e,p_v,Y)$ for any $t>0$. By Assumptions (1-4) the degree distribution of $H$ follows a power-law with an exponential cutoff, i.e.,
	\[
		\E\left[\frac{N_{k,t}}{|V_t|}\right] \sim c \cdot k^{-\beta} \gamma^k \left(\frac{1}{k} + \delta \right)
	\]
	%for
	\[
	\begin{split}
	for \quad \beta & = \frac{\mu(p_v+p_e)-p_d \theta}{p_v(\mu-1)+p_e \mu + p_d},	\quad \quad \gamma  = \frac{p_v(\mu-1)+p_e \mu}{p_v(\mu-1)+p_e\mu+p_d}, \\
	\delta & = \frac{p_d}{\mu(p_v+p_e)-p_d \theta}, \quad \quad \quad \quad c  = \frac{\beta \cdot \Gamma(1+\beta)}{\gamma},
	\end{split}
	\]
	where $\Gamma(x)$ stands for the gamma function \textnormal{(}$\Gamma(x) = \int_{0}^{\infty} t^{x-1}e^{-t} \,\mathrm{d}t$\textnormal{)}.
\end{theorem}

\begin{remark*}
	The theorem and its proof presented below remain true if we relax Assumption~(3) just to $D_t = \E[D_t] + o(t)$ whp. Nevertheless, we leave the stronger version of (3) on the list of assumptions as it will be needed in Section \ref{sec:theta} for estimating~$\theta$.
\end{remark*}

\begin{remark*}
	Setting $p_d=0$ in the above theorem (i.e., considering the process without deactivation) results in the power-law degree distribution, namely	$\E\left[\frac{N_{k,t}}{|V_t|}\right] \sim c \cdot k^{-(\beta+1)}$, where $\beta = \frac{\mu}{\mu-p_v}$ and $c = \beta \cdot \Gamma(1+\beta)$. This is in accordance with the result obtained in~\cite{ALP_hyper}.
\end{remark*}

%Below we present a proof with some shortcuts in calculations. The full proof can be found in the Appendix.

The proof below contains shortcuts in calculations. See the Appendix A for the full proof.
\vspace{5pt}
\begin{proof}
	We take a standard master equation approach that can be found e.g. in Chung and Lu book \cite{ChLu_book} about complex networks. %or Avin et al. paper \cite{ALP_hyper} on preferential attachment hypergraphs.
	However, we apply it separately to the number of active vertices and the number of deactivated vertices.
	
	Recall that $N_{k,t}$ denotes the number of vertices of degree $k$ at time $t$. We need to show that $\lim_{t \rightarrow \infty} \E\left[\frac{N_{k,t}}{|V_t|}\right] \sim c \cdot k^{-\beta} \gamma^k \left(\frac{1}{k} + \delta \right)$ 
%	\[
%	\lim_{t \rightarrow \infty} \E\left[\frac{N_{k,t}}{|V_t|}\right] \sim c \cdot k^{-\beta} \gamma^k \left(\frac{1}{k} + \delta \right)
%	\]
	for the proper constants $c, \beta, \gamma$ and $\delta$.
	However, by Lemma \ref{lemma:Nkt_limit} we know that it suffices to show that
	\[
	\lim_{t \rightarrow \infty} \frac{\E[N_{k,t}]}{t} \sim p_v \cdot c \cdot k^{-\beta} \gamma^k \left(\frac{1}{k} + \delta \right).
	\]
	
	Recall that $N_{k,t} = A_{k,t} + I_{k,t}$. First, let us evaluate $\lim_{t \rightarrow \infty} \frac{\E[A_{k,t}]}{t}$ using the mathematical induction on $k$. In this part we follow closely the lines of the proof that can be found in~\cite{ALP_hyper}. Consider the case $k=1$. Since $H_0$ consists of a single hyperedge of cardinality 1 over
	a single vertex, we have $A_{1,0} = 1$. To formulate a master equation, let us make the following observation for $t \geq 1$. An active vertex remains in $\mathcal{A}_{1,t}$ if it had degree 1 at step $t - 1$ and was neither selected to a hyperedge, nor deactivated. Recall that a vertex from $\mathcal{A}_{1,t-1}$ is chosen at step $t$ in a single trial to the new hyperedge with probability $1/D_{t-1}$ thus the chance that it won't be selected to the hyperedge of cardinality $y$ equals $(1-1/D_{t-1})^y$.
	Also, in each step, with probability $p_v$, a single new active vertex of degree $1$ is added to the hypergraph.
	Let $\mathcal{F}_t$ denote a $\sigma$-algebra associated with the probability space at step $t$. 
	For $t \geq 1$ we have
	\begin{equation} \label{equation:E[A_(1,t)]-recurrence}
		%\begin{split}
		\E[A_{1,t} | \mathcal{F}_{t-1}] 
		= p_v A_{1,t-1} \left(1 - \frac{1}{D_{t-1}}\right)^{Y_t - 1}
		 + p_e A_{1,t-1} \left(1 - \frac{1}{D_{t-1}}\right)^{Y_t} 
		+ p_d A_{1,t-1} \left(1 - \frac{1}{D_{t-1}}\right)
		+ p_v. 
		%\end{split}
	\end{equation}
	After taking the expectation on both sides of (\ref{equation:E[A_(1,t)]-recurrence}) we derive upper and lower bounds on $\E[A_{1,t}]$. By Bernoulli's inequality ($(1+x)^n \geq 1+nx$ for $n \in \mathbb{N}$ and $x\geq-1$), Lemma \ref{lemma:A/D} (thus by Assumption (3) necessary for it) and the independence of $Y_t$ from $A_{1,t-1}$ and $D_{t-1}$ %we obtain
	\begin{equation} \label{equation:E[A_(1,t)]-lower-bound}
    \begin{split} 
		\E[A_{1,t}] &
		\geq p_v \E\left[A_{1,t-1} \left({1 - \frac{Y_t - 1}{D_{t-1}}}\right)\right] 
		+ p_e \E\left[A_{1,t-1} \left({1 - \frac{Y_t}{D_{t-1}}}\right)\right] \\
		& \quad + p_d \E\left[{A_{1,t-1} \left({1 - \frac{1}{D_{t-1}}}\right)}\right]
		+ p_v \\
		%&= p_v \E[{A_{1,t-1}}] \left({1 - \frac{\E[{Y_t}] - 1}{\E[{D_{t-1}}]}}\right)
		%+ p_e \E[{A_{1,t-1}}] \left({1 - \frac{\E[Y_t]    %}{\E[{D_{t-1}}]}}\right) \\
		%&\quad
		%+ p_d \E[{A_{1,t-1}}] \left({1 - \frac{1}{\E[{D_{t-1}}]}}\right)
		%+ p_v - o(1) \\
		&= \E[A_{1,t-1}] \left({
			1 - \frac{p_v (\mu - 1) + p_e \mu + p_d}{\E[{D_{t-1}}]}}\right) + p_v - o(1). 
	\end{split}
	\end{equation}
	On the other hand, since $(1 - x)^n \leq 1/(1 + nx)$ for $x\in[0, 1]$ and $n\in\mathbb{N}$, and $A_{1,t-1} \leq t$, by Lemma \ref{lemma:Y2/D2} (thus by Assumptions (1) and (2) necessary for it) and Lemma \ref{lemma:A/D} (thus by Assumption (3)) we have
	\begin{equation} \label{equation:E[A_(1,t)]-upper-bound}
	\begin{split}
		\E[{A_{1,t}}]
		&\leq p_v \E\left[{\frac{A_{1,t-1}}{1 + (Y_t - 1)/D_{t-1}}}\right]
		+ p_e \E\left[{\frac{A_{1,t-1}}{1 +  Y_t     /D_{t-1}}}\right] 
		+ p_d\E\left[{A_{1,t-1} \left({1 - \frac{1}{D_{t-1}}}\right)}\right]
		+ p_v \\
		%
%		&=   p_v \E\left[{A_{1,t-1} \left({1 - \frac{Y_t - 1}{D_{t-1} + Y_t - 1}}\right)}\right]
%		+ p_e \E\left[{A_{1,t-1} \left({1 - \frac{Y_t    }{D_{t-1} + Y_t}}\right)}\right] \\
%		&\quad
%		+ p_d \E\left[{A_{1,t-1} \left({1 - \frac{1}{D_{t-1}}}\right)}\right]
%		+ p_v \\
%		%
%		&\leq p_v \E\left[{A_{1,t-1} \left({1 - \frac{Y_t - 1}{D_{t-1}} + \frac{(Y_t - 1)^2}{D_{t-1}^2}}\right)}\right]
%		+ p_e \E\left[{A_{1,t-1} \left({1 - \frac{Y_t    }{D_{t-1}} + \frac{ Y_t     ^2}{D_{t-1}^2}}\right)}\right] \\
%		&\quad
%		+ p_d \E\left[{A_{1,t-1} \left({1 - \frac{1}{D_{t-1}}}\right)}\right]
%		+ p_v \\
%		%
%		&= p_v \E[{A_{1,t-1}}] \left({1 - \frac{\E[{Y_t}] - 1}{\E[{D_{t-1}}]}}\right)
%		+ p_e \E[{A_{1,t-1}}] \left({1 - \frac{\E[{Y_t}]    }{\E[{D_{t-1}}]}}\right) \\
%		&\quad
%		+ p_d \E[{A_{1,t-1}}] \left({1 - \frac{1}{\E[{D_{t-1}}]}}\right)
%		+ p_v 
%		+ p_v \E\left[{O(t) \frac{(Y_t - 1)^2}{D_{t-1}^2}}\right]
%		+ p_e \E\left[{O(t) \frac{Y_t^2}{D_{t-1}^2}}\right] + o(1)\\
%		%
		&= \E[{A_{1,t-1}}] \left({
			1 - \frac{p_v (\mu - 1) + p_e \mu + p_d}{\E[{D_{t-1}}]}
		}\right) + p_v + o(1). 
	\end{split}
	\end{equation}
	From (\ref{equation:E[A_(1,t)]-lower-bound}) and (\ref{equation:E[A_(1,t)]-upper-bound}) we get
	\[
		\E[{A_{1,t}}] = \E[{A_{1,t-1}}] \left({
			1 - \frac{p_v (\mu - 1) + p_e \mu + p_d}{\E[{D_{t-1}}]}
		}\right) + p_v + o(1). 
	\]
	Now, we apply Lemma \ref{lemma:rec_seq} to the above equation choosing
		\[
		a_t  = \E[{A_{1,t}}], \quad
		b_t  =  \frac{p_v (\mu - 1) + p_e \mu + p_d}{\E[{D_{t-1}}]/t}, \quad
		c_t  =  p_v + o(1).
		\]
	We have $\lim_{t\to\infty} c_t = p_v$ and, by Fact \ref{fact:Dt/t} (thus by Assumptions (1) and (4) implying it), $\lim_{t\to\infty} b_t = \frac{p_v (\mu-1) + p_e \mu + p_d}
	{\mu (p_v + p_e) - p_d \theta} =: 1/\beta$ thus 
	\[
	 \lim_{t\to\infty} \frac{a_t}{t} = \lim_{t\to\infty} \frac{\E[{A_{1,t}}]}{t} = \frac{p_v}{1 + 1/\beta} =: \bar{A}_1.
	\]
	Now, we assume that the limit $\lim_{t\to\infty} \frac{\E[{A_{k-1,t}}]}{t}$ exists and equals $\bar{A}_{k-1}$ and we will show by induction on $k$ that the analogous limit for $\E[{A_{k,t}}]$ exists. Let us again formulate a master equation, this time for $k>1$. We have $A_{k,0}=0$ and for $t \geq 1$ an active vertex appears in $A_{k,t}$ if it was active at step $t-1$, had degree $k-l$ and was chosen exactly $l$ times to a hyperedge, or it had degree $k$ and was not selected for deactivation. Let $B({l,n,p}) = \binom{n}{l} p^l (1-p)^{n-l}$. We have
	\[
		\begin{split}
			\E[{A_{k,t} | \mathcal{F}_{t-1}} ]
			& = p_v \smashoperator{\sum_{l=0}^{\min\{Y_t-1, k-1\}}} 
			A_{k-l,t-1} B\left({l, Y_t - 1, \frac{k-l}{D_{t-1}}}\right)
			 + p_e \smashoperator{\sum_{l=0}^{\min\{Y_t, k-1\}}} 
			A_{k-l,t-1} B\left({l, Y_t , \frac{k-l}{D_{t-1}}}\right) \\
			& \quad + p_d A_{k,t-1} \left({1 - \frac{k}{D_{t-1}}}\right).
		\end{split}
	\]
	%We will now extract two first summands from each sum and group them together with the last term of the above expression since (as we prove below) they compose the significant term.
	Taking the expectation on both sides we get
	\[
		\E[{A_{k,t}}] = 
		\E[\psi] + p_v \E[{\varphi(Y_t - 1)}] + p_e \E[{\varphi(Y_t)}], %\quad \textnormal{where}
	\]
	where
	\[
	\begin{split}
		\psi 
		&= p_v \sum_{l=0}^1
		A_{k-l,t-1} B\left({l, Y_t - 1, \frac{k-l}{D_{t-1}}}\right)
		 + p_e \sum_{l=0}^1 
		A_{k-l,t-1} B\left({l, Y_t, \frac{k-l}{D_{t-1}}} \right) 
		 + p_d A_{k,t-1} \left({1 - \frac{k}{D_{t-1}}}\right) \\
%		&= A_{k,t-1} \left({
%			p_v \left({1 - \frac{k}{D_{t-1}}}\right)^{Y_t-1} + 
%			p_e \left({1 - \frac{k}{D_{t-1}}}\right)^{Y_t  } +
%			p_d \left({1 - \frac{k}{D_{t-1}}}\right)
%		}\right) \\
%		&\quad
%		+ A_{k-1,t-1} \frac{k-1}{D_{t-1}} \left({
%			p_v (Y_t - 1) \left({1 - \frac{k-1}{D_{t-1}}}\right)^{Y_t-2} + 
%			p_e  Y_t     \left( {1 - \frac{k-1}{D_{t-1}}}\right)^{Y_t-1}
%		}\right)
	\end{split}
	\]
	and 
	\[\quad \varphi(n) = \smashoperator{\sum_{l=2}^{\min\{n, k-1\}}} A_{k-l,t-1} B\left({l, n, \frac{k-l}{D_{t-1}}}\right).\]
%	\begin{equation}
%	\begin{split}
%		\varphi(n) = \smashoperator{\sum_{l=2}^{\min\{n, k-1\}}} A_{k-l,t-1} B\left({l, n, \frac{k-l}{D_{t-1}}}\right).
%	\end{split}
%	\end{equation}

	\noindent
	We will show that only the term $\E[\psi]$ is significant and that the terms $\E[{\varphi(Y_t - 1)}]$ and $\E[{\varphi(Y_t)}]$ converge to 0 as $t\to\infty$. We have
	\[
	 \begin{split}
		\varphi(Y_t) 
		&\leq \sum_{l=2}^{k-1} A_{k-l,t-1} 
		\binom{Y_t}{l} \left({\frac{k - l}{D_{t-1}}}\right)^l \left({1 - \frac{k-l}{D_{t-1}}}\right)^{Y_t - l}  
		 = O(t) \frac{Y_t^2}{D_{t-1}^2}.
%		&\leq O(t) \sum_{l=2}^{k-1} 
%		\binom{Y_t}{l} \left({\frac{k - l}{D_{t-1}}}\right)^l \left({1 - \frac{k-l}{D_{t-1}}}\right)^{Y_t - l} \\
%		&\leq O(t) \sum_{l=2}^{k-1} 
%		Y_t^l \left({\frac{k}{D_{t-1}}}\right)^l\left( {1 - \frac{1}{D_{t-1}}}\right)^{Y_t - k + 1} \\
%		&\leq O(t)
%		\frac{Y_t^2 k^2}{D_{t-1}^2} e^{-Y_t/D_{t-1}} e^{k-1} \sum_{l=2}^{k-1} \left({\frac{Y_t k}{D_{t-1}}}\right)^{l-2} \\
%		&= O(t)
%		\frac{Y_t^2}{D_{t-1}^2} e^{-Y_t/D_{t-1}} \sum_{l=2}^{k-1} \left({\frac{Y_t k}{D_{t-1}}}\right)^{l-2} = O(t) \frac{Y_t^2}{D_{t-1}^2}.
	\end{split}
	\]
%	Then, if $Y_t \leq D_{t-1}$, we have 
%	\[
%		\varphi(Y_t) \leq 
%		O(t) \frac{Y_t^2}{D_{t-1}^2} k^{k-2} =
%		O(t) \frac{Y_t^2}{D_{t-1}^2}.
%	\]
%	Otherwise, 
%	\[
%		\begin{split}
%			\varphi(Y_t) 
%			&\leq O(t) \frac{Y_t^2}{D_{t-1}^2} e^{-Y_t/D_{t-1}} 
%			\frac{(Y_t k / D_{t-1})^{k-2} - 1}{(Y_t k / D_{t-1}) - 1} = O(t) \frac{Y_t^2}{D_{t-1}^2}. \\
%			%& \leq O(t) \frac{Y_t^2}{D_{t-1}^2} e^{-Y_t/D_{t-1}} 
%			%\frac{(Y_t   / D_{t-1})^{k-2}}{k - 1}  k^{k-2} \\
%%			&\leq O(t) \frac{Y_t^2}{D_{t-1}^2} e^{-(k-2)}
%%			\frac{(k - 2)^{k-2}}{k - 1}  k^{k-2} = O(t) \frac{Y_t^2}{D_{t-1}^2}.
%		\end{split}
%	\]
	%where the last inequality follows from the fact that $e^{-x} x^{\alpha}$ is maximized at $x = \alpha$.
	Hence by Lemma \ref{lemma:Y2/D2} (thus by Assumptions (1) and (2)) %in both above cases
	we get $\E[\varphi(Y_t)] = o(1)$ and, similarly, $\E[\varphi(Y_t-1)] = o(1)$. %Now, we derive the bounds for $\E[\psi]$ analogous to the ones derived for $\E[A_{1,t}]$.
	The bounds for $\E[\psi]$ can be derived analogously to the ones for $\E[A_{1,t}]$ and they give
	\begin{equation}
	\begin{split}
		\E[{A_{k,t}}]
		&= \E[{A_{k,t-1}}] \left({
			1 - \frac{k \left({p_v (\mu-1) + p_e \mu + p_d}\right)}{\E[{D_{t-1}}]}
		}\right) \\
	& \quad + \E[{A_{k-1,t-1}}] \frac{(k-1) \left({p_v (\mu-1) + p_e \mu}\right)}{\E[{D_{t-1}}]} + o(1).
	\end{split}
	\end{equation}
	Recall that by the induction assumption $\lim_{t\to\infty} \E[{A_{k-1,t}}]/t = \bar{A}_{k-1}$. Now, we apply again Lemma \ref{lemma:rec_seq} to the above equation choosing
	\[
	\begin{split}
		& a_t = \E[{A_{k,t}}], \quad
		b_t = \frac{
			k (p_v (\mu-1) + p_e \mu + p_d)
		}{
			\E[{D_{t-1}}]/t
		},  \quad
		 c_t = \frac{\E[{A_{k-1,t-1}}]}{t}
		\frac{
			(k - 1) (p_v (\mu-1) + p_e \mu )
		}{
			\E[{D_{t-1}}]/t
		}
		+ o(1).
	\end{split}
	\]
	By Fact \ref{fact:Dt/t} (thus by Assumptions (1) and (4)) we have $\lim_{t\to\infty} b_t = k / \beta$ and $\lim_{t\to\infty} c_t = \bar{A}_{k-1} \frac{(k - 1) (p_v (\mu-1) + p_e \mu)}{\mu (p_v + p_e) - p_d \theta}$
%	\[
%		 \lim_{t\to\infty} b_t = k / \beta \quad \textnormal{and} \quad 
%		 \lim_{t\to\infty} c_t = \bar{A}_{k-1}
%		\frac{
%			(k - 1) (p_v (\mu-1) + p_e \mu)
%		}{
%			\mu (p_v + p_e) - p_d \theta
%		}
%	\]
	thus
	\begin{equation}
		\lim_{t\to\infty} \frac{a_t}{t}
		= \lim_{t\to\infty} \frac{\E[{A_{k,t}}]}{t}
		= \bar{A}_k
		= \bar{A}_{k-1} \frac{
			(k - 1) \gamma
		}{
			k + \beta
		},
	\end{equation}
	where $\gamma = \frac{p_v (\mu - 1) + p_e \mu}{p_v (\mu - 1) + p_e \mu + p_d}$. 
%	\[%\label{equation:gamma}
%		\gamma = \frac{p_v (\mu - 1) + p_e \mu}{p_v (\mu - 1) + p_e \mu + p_d}.
%	\]
	Thus we got
	\[
	\begin{split}
		& \bar{A}_1 =  p_v \beta \frac{1}{(1 + \beta)}, \quad 
		\bar{A}_2 = p_v \beta \frac{\gamma}{(1 + \beta)(2 + \beta)}, \ldots,  \quad 
		 \bar{A}_k = p_v \beta \frac{\gamma^{k-1} (k-1)!}{(1 + \beta)(2 + \beta) \dots (k + \beta)}.
	\end{split}
	\]
	Since $\lim_{k \rightarrow \infty} \frac{\Gamma(k)k^{\alpha}}{\Gamma(k+\alpha)} = 1$ for constant $\alpha \in \mathbb{R}$ we have
	\begin{equation} \label{eq:Ak}
		\begin{split}
		\lim_{t\to\infty}& \frac{\E[{A_{k,t}}]}{t} = \bar{A}_k =
		\frac{p_v \beta}{\gamma} 
		\frac{\gamma^k \Gamma(1 + \beta) \Gamma(k)}
		{\Gamma(k + \beta + 1)}
	 \sim p_v \cdot c \cdot \gamma^k k^{-(\beta + 1)} %\quad \textnormal{with } c = \frac{\beta \cdot \Gamma(1 + \beta)}{\gamma}.
		\end{split}
	\end{equation}
	with $c = \frac{\beta \cdot \Gamma(1 + \beta)}{\gamma}$.
%	\[
%		c = \frac{\beta \cdot \Gamma(1 + \beta)}{\gamma}.% = \frac{\mu(p_v + p_e)-p_d \theta}{p_v(\mu-1)+p_e\mu} \cdot \Gamma(1+ \beta).
%	\]

	Now, let us evaluate $\lim_{t\to\infty} \frac{\E[{I_{k,t}}]}{t}$. We have $I_{k,0} = 0$ for all $k \geq 1$. For $t \geq 1$ the expected number of inactive vertices of degree $k \geq 1$ at step $t$, given $\mathcal{F}_{t-1}$, can be expressed as
	\[
		\E[{I_{k,t} | \mathcal{F}_{t-1}}] =
		I_{k,t-1} + p_d A_{k,t-1} \frac{k}{D_{t-1}},
	\]
	since inactive vertices of degree $k$ remain in $I_{k,t}$ forever and a vertex of degree $k$ becomes inactive if it was selected in step $t-1$ for deactivation.
	Taking the expectation on both sides, by Lemma \ref{lemma:A/D} (thus by Assumption (3)), we obtain
	\[
		\E[{I_{k,t}}] = 
		\E[{I_{k,t-1}} ]
		+ p_d \E[{A_{k,t-1}}] \frac{k}{\E[{D_{t-1}}]} + o(1).
	\]
	Then, by Fact \ref{fact:Dt/t} (thus by Assumptions (1) and (4)),
	\[
	\begin{split}
		\lim_{t\to\infty} & \left(\E[I_{k,t}] - \E[{I_{k,t-1}}]\right)
		= \lim_{t\to\infty} p_d k
		\frac{\E[{A_{k,t-1}}]\cdot t}{t \cdot \E[{D_{t-1}}]}
		 + o(1) 
		  = \bar{A}_k k \delta, %\quad \textnormal{where } \delta = \frac{p_d}{(p_v + p_e) \mu - p_d \theta}.
%		&= \bar{A}_k \frac{p_d k}{(p_v + p_e) \mu - p_d \theta} = \bar{A}_k k \delta,
	\end{split}
	\]
	where $\delta = \frac{p_d}{(p_v + p_e) \mu - p_d \theta}$. 
%	\[
%		\delta = \frac{p_d}{(p_v + p_e) \mu - p_d \theta}.
%	\]
	And, by Stolz--Ces\`aro theorem (Theorem \ref{thm:S-C}), we obtain
	\begin{equation} \label{eq:Ik}
		\bar{I}_k := 
		\lim_{t\to\infty} \frac{\E[{I_{k,t}}]}{t} = 
		\lim_{t\to\infty} (\E[{I_{k,t}}] - \E[{I_{k,t-1}}]) = 
		\bar{A}_k k \delta.
	\end{equation}
	Finally, by (\ref{eq:Ak}) and (\ref{eq:Ik})
	\[
	\begin{split}
		\lim_{t\to\infty} & \frac{\E[{N_{k,t}}]}{t} = \lim_{t\to\infty} \frac{\E[{A_{k,t}}]+\E[{I_{k,t}}]}{t}  = 
		\bar{A}_k + \bar{I}_k 
		 = \bar{A}_k (1 + k\delta) \sim
		p_v \cdot c \cdot k^{-\beta} \gamma^k  \left({\frac{1}{k} + \delta}\right).
	\end{split}
	\]
%    and, by Lemma \ref{lemma:Nkt_limit},
%	\[
%	\lim_{t\to\infty} \frac{\E[{N_{k,t}}]}{|V_t|} \sim c \cdot   k^{-\beta}\gamma^k \left({\frac{1}{k} + \delta}\right).
%	\]
\end{proof}

\section{Estimating the limiting value $\theta$} \label{sec:theta}

This section is devoted to estimating $\theta$ which appears as one of the parameters in the degree distribution of our hypergraph model $H$ (consult Theorem \ref{thm:deg_dist}). Recall that $\Theta_t$ stands for the degree of a vertex chosen for deactivation at time $t$ and it appears in the fourth assumption needed to prove Theorem \ref{thm:deg_dist} \hspace{3pt} \textbf{4.}~$\lim_{t \rightarrow \infty} \frac{1}{t} \sum_{\tau=1}^{t} \E[\Theta_{\tau}] = \theta \in \mathbb{R}_{>0}$.
%\begin{enumerate}
%\setcounter{enumi}{3}
%\item $\lim_{t \rightarrow \infty} \frac{1}{t} \sum_{\tau=1}^{t} \E[\Theta_{\tau}] = \theta \in \mathbb{R}_{>0}$.
%\end{enumerate}

Let us start with showing that $\sum_{\tau=1}^{t} \E[\Theta_{\tau}]$ is of order $\Theta(t)$.

\begin{lemma} \label{lemma:linear_t}
	%Let $\Theta_t$ denote the degree of a vertex chosen for deactivation at time $t$ and
	Assume that $\E[Y_t] = \mu$ for all $t>0$. Then %$p_d~\leq~\frac{1}{t}\sum_{\tau=1}^{t} \E[\Theta_{\tau}] \leq 1 + \frac{p_v(\mu-1)+p_e\mu}{p_d}$.
	\[
		p_d \leq \frac{1}{t}\sum_{\tau=1}^{t} \E[\Theta_{\tau}] \leq 1 + \frac{p_v(\mu-1)+p_e\mu}{p_d}.
	\]
\end{lemma}
\begin{proof}
	By equation (\ref{eq:Dt_Theta}) we get $\sum_{\tau=1}^{t}\E[\Theta_{\tau}] = \frac{1}{p_d}\left(1 + (p_v+p_e) \mu t - \E[D_t]\right)$.
%	\[
%		\sum_{\tau=1}^{t}\E[\Theta_{\tau}] = \frac{1}{p_d}\left(1 + (p_v+p_e) \mu t - \E[D_t]\right).
%	\]
	Note that $\E[D_t] \geq \E[A_t] = 1 + (p_v-p_d)t$ (we assume $p_v>p_d$) thus on one hand $\sum_{\tau=1}^{t}\E[\Theta_{\tau}] \leq t \left(1 + \frac{p_v(\mu-1)+p_e\mu}{p_d} \right)$
%	\[
%		\sum_{\tau=1}^{t}\E[\Theta_{\tau}] \leq t \left(1 + \frac{p_v(\mu-1)+p_e\mu}{p_d} \right)
%	\]
	and on the other $\sum_{\tau=1}^{t}\E[\Theta_{\tau}] \geq \E[I_t] = p_d t$.
%	\[
%		\sum_{\tau=1}^{t}\E[\Theta_{\tau}] \geq \E[I_t] = p_d t.
%	\]
\end{proof}

Unfortunately, we were not able to prove that the limit $\lim_{t \rightarrow \infty} \frac{1}{t} \sum_{\tau=1}^{t} \E[\Theta_{\tau}]$ exists. However, we support this assumption by simulations in Section \ref{sec:experiments}. Whereas in this section we show, assuming that the limit exists, how to estimate it.

Throughout this section $F(a,b;c;z)$ denotes the Gaussian hypergeometric function, i.e., for $a,b,c,z \in \mathbb{C}$, $|z|<1$
\[
F(a,b;c;z) = \sum_{n=0}^{\infty} 
\frac{(a)_n (b)_n}{(c)_n} \frac{z^n}{n!}, %\quad \quad \textnormal{for} \quad a,b,c,z \in \mathbb{C}, \quad |z|<1 \quad \textnormal{and} \quad c \notin \mathbb{Z}_{\leq 0},
\]
where $(x)_n = \Gamma(x+n)/\Gamma(x)$. %and $\Gamma(x) = \int_{0}^{\infty} t^{x-1} e^{-t} \dt$ is the gamma function. 
%is the Pochhammer symbol.

\begin{lemma} \label{lemma:k2Akt/Dt}
	Assume that $\Pr[D_t \neq \E[D_t] + o(t)] = o(1/t)$. Then %$\E\left[ \frac{\sum_{k \geq 1} k^2 A_{k,t}}{D_t}\right] = \frac{\E[\sum_{k \geq 1} k^2 A_{k,t}]}{\E[D_t]} + o(1)$.
	\[
		\E\left[ \frac{\sum_{k \geq 1} k^2 A_{k,t}}{D_t}\right] = \frac{\E[\sum_{k \geq 1} k^2 A_{k,t}]}{\E[D_t]} + o(1).
	\]
\end{lemma}
The proof can be found in the Appendix B.

\begin{theorem} \label{thm:R_fixed}
	Assume that the conditions $(1-4)$ from Section \ref{sec:model} hold. %Assume that  $\lim_{t \rightarrow \infty} \frac{1}{t} \sum_{\tau=1}^{t} \E[\Theta_{\tau}] = \theta \in \mathbb{R}_{>0}$.
	Then $\theta$ is a fixed point of the function $R(x) := \frac{F(2,2;\rho(x); \gamma)}{F(1,2;\rho(x);\gamma)}$,
%	\[
%		R(x) := \frac{F(2,2,\rho(x), \gamma)}{F(1,2,\rho(x),\gamma)}, 
%	\]
	where $\gamma = \frac{p_v(\mu-1)+p_e \mu}{p_v(\mu-1)+p_e\mu+p_d}$ and $\rho(x) = 2 + \frac{\mu(p_v+p_e)-p_d x}{p_v(\mu-1)+p_e \mu + p_d}$.
\end{theorem}
\begin{proof}
	Recall that $\theta = \lim_{t \rightarrow \infty} \frac{1}{t} \sum_{\tau=1}^{t} \E[\Theta_{\tau}]$ and $\Theta_t$ is the degree of a vertex chosen for deactivation at time $t$. Let $\mathcal{F}_t$ denote a $\sigma$-algebra associated with the probability space at step $t$. We have $\E[\Theta_t | \mathcal{F}_{t-1}] = \sum_{k \geq 1} k \frac{k A_{k,{t-1}}}{D_{t-1}}$, 
%	\[
%	\E[\Theta_t | \mathcal{F}_{t-1}] = \sum_{k \geq 1} k \frac{k A_{k,{t-1}}}{D_{t-1}},
%	\]
	hence taking expectation on both sides, applying Lemma \ref{lemma:k2Akt/Dt} (thus by Assumption (3)) and noting that $D_t = \sum_{k \geq 1} k A_{k,t}$ we get
	\[
		\begin{split}
		\E[\Theta_t] & = \E\left[ \frac{\sum_{k \geq 1} k^2 A_{k,t-1}}{D_{t-1}}\right] = \frac{\E[\sum_{k \geq 1} k^2 A_{k,t-1}]}{\E[D_{t-1}]} + o(1) 
		 = \frac{\E[\sum_{k \geq 1} k^2 A_{k,t-1}]}{\E[\sum_{k \geq 1} k A_{k,t-1}]} + o(1).
		\end{split}
	\]
	Now, by equation (\ref{eq:Ak}) (thus by Assumptions (1-4) needed to prove Theorem \ref{thm:deg_dist}) we write
	\[
		\lim_{t \rightarrow \infty} \E[\Theta_t] = \frac{\sum_{k \geq 1} k^2 \bar{A}_{k}}{\sum_{k \geq 1} k \bar{A}_k} = \frac{F(2,2;\rho(\theta); \gamma)}{F(1,2;\rho(\theta);\gamma)}.
	\]
	Finally, setting $a_t = \sum_{\tau = 1}^{t} \E[\Theta_{\tau}]$ and $b_t = t$ in Stolz-Ces{\`a}ro theorem (Theorem \ref{thm:S-C}) we obtain
	\[
		\theta = \lim_{t \rightarrow \infty} \frac{1}{t} \sum_{\tau=1}^{t} \E[\Theta_{\tau}] = \frac{F(2,2;\rho(\theta); \gamma)}{F(1,2;\rho(\theta);\gamma)}.
	\]
\end{proof}

From now on we consider the behavior of $R(x)$ only in the interval $[0,\hat{\theta}]$, where $\hat{\theta}~=~\frac{(p_v+p_e)\mu}{p_d}$ since we know that the limiting value $\theta$ we are looking for belongs there. Indeed, by Lemma~\ref{lemma:linear_t} we know that it is at least $p_d$ and at most $\frac{(p_v+p_e)\mu -p_v + p_d}{p_d}$ and we work by $p_v > p_d$ to ensure that, on average, we add more vertices to the network than we deactivate. Recall that the function $F(a, b; c; z)$ is defined for $|z| < 1$ and $c \not\in \mathbb{Z}_{\leq 0}$.
Therefore, since $0 < \gamma < 1$ and $\rho(x)$ is positive on $[0, \hat\theta]$, both $F(1,2;\rho(x);\gamma)$ and $F(2,2;\rho(x);\gamma)$ are always defined, continuous and positive on $[0, \hat\theta]$. This implies that $R(x)$ is continuous on $[0,\hat\theta]$. Below we will justify that $R(x)$ has just one fixed point in the interval $[0,\hat{\theta}]$ and that a fixed-point iteration method will converge here. We start with recalling Banach Fixed Point Theorem.

\begin{theorem}[Banach Fixed Point Theorem] \label{thm:Banach}
	Let $(S,d)$ be a non-empty complete metric space with a contraction mapping $R:S \rightarrow S$. Then $R$ admits a unique fixed point $s^*$ in $S$ ($R(s^*)=s^*$). Furthermore, $s^*$ can be found as follows: start with an arbitrary element $s_0 \in S$ and define a sequence $\{s_n\}_{n \geq 1}$ by $s_n = R(s_{n-1})$ for $n \geq 1$. Then $\lim_{n \rightarrow \infty} s_n = s^*$.
\end{theorem}

Thus we aim at showing that $R(x)$ is a contraction mapping on $[0,\hat{\theta}]$. From now on let $F_1(x) = F(1,2;\rho(x);\gamma)$ and $F_2(x) = F(2,2;\rho(x);\gamma)$ for $\gamma$ and $\rho(x)$ as in Theorem \ref{thm:R_fixed}.

The proofs of Lemmas \ref{lemma:r_formula}, \ref{lemma:R-increases}, and \ref{lemma:R_contr_map} can be found in the Appendix B.

 \begin{lemma} \label{lemma:r_formula}
 	The function $R(x) = \frac{F_2(x)}{F_1(x)}$ can be also expressed as $R(x) = x - \frac{p_v}{p_d} + \frac{1}{1 - \gamma} \frac{\rho(x) - 1}{F_1(x)}$,
%	\[%\label{equation:R-new}
%		R(x)
%		= x - \frac{p_v}{p_d} + 
%		\frac{1}{1 - \gamma}
%		\frac{\rho(x) - 1}{F_1(x)},
%	\]
	where $\gamma$ and $\rho(x)$ are as in Theorem \ref{thm:R_fixed}.
\end{lemma}

\begin{lemma}\label{lemma:R-increases}
	The function $R(x)$ strictly increases on $[0, \hat\theta]$.
\end{lemma}

\begin{lemma} \label{lemma:R_contr_map}
	The function $R(x)$ is a contraction mapping on $[0, \hat\theta]$.
\end{lemma}
%\begin{proof}
%	Recall that a function $f : S \mapsto S$, defined on a metric space $(S, d)$, is called a contraction mapping, if there exists a constant $q \in [0, 1)$, such that for all $s_1, s_2 \in S$, we have $d(f(s_1), f(s_2)) \leq q d(s_1, s_2)$.
%%	\[
%%		d(f(s_1), f(s_2)) \leq q d(s_1, s_2).
%%	\]
%	If $f(s)$ is a differentiable function, such that $\sup\,|f'(s)| < 1$, then $f(s)$ is a contraction mapping with $q = \sup\,|f'(s)|$.
%	
%	Using the form of $R(x)$ presented in Lemma \ref{lemma:r_formula}, we obtain
%	\[
%		R'(x) = 
%		1 + \frac{1}{1 - \gamma} \frac{ \rho'(x)      F_1 (x) - 
%			(\rho (x) - 1) F_1'(x)}	{F_1(x)^2}.
%	\]
%	$F_1(x)$ is positive and increases, and $\rho(x) > 1$ and decreases on $[0, \hat\theta]$, which implies that the right term of the expression is negative.
%	Since $R(x)$ also increases on $[0, \hat\theta]$ (Lemma \ref{lemma:R-increases}), we have that $|R'(\theta)| \in [0, 1)$ for any $\theta \in [0, \hat\theta]$.
%	Therefore, by the extreme value theorem, we know that $|R'(\theta)|$ achieves some maximum value $q \in (0, 1)$.
%	Then, since $[0, \hat\theta]$ is a complete metric space and $R([0, \hat\theta]) \subseteq [0, \hat\theta]$ (using the formula from Lemma \ref{lemma:r_formula} it is easy to check that $R(0)>0$ and $R(\hat\theta) < \hat\theta$), we conclude that $R(x)$ is a contraction mapping on $[0, \hat\theta]$.
%\end{proof}

\begin{corollary} \label{cor:theta_est}
	Assume that the conditions $(1-4)$ from Section~\ref{sec:model} hold \textnormal{(}in particular, $\theta =  \lim_{t \rightarrow \infty} \frac{1}{t} \sum_{\tau=1}^{t} \E[\Theta_{\tau}]$\textnormal{)}. Then $\theta$ is a unique fixed point of $R(x)$ in $[0, \hat\theta]$, such that $\lim_{n\to\infty} \theta_n = \theta$, 
	%\begin{equation}
	%	\lim_{n\to\infty} \theta_n = \theta,
	%\end{equation}
	where $\theta_{n+1} = R(\theta_n)$ and $\theta_0$ can take any value in $[0, \hat\theta]$.
\end{corollary}
\begin{proof}
	The proof follows directly from the fact that $\theta$ is a fixed point of $R(x)$ (Theorem \ref{thm:R_fixed}), the fact that $R(x)$ is a contraction mapping defined on a complete metric space (Lemma~\ref{lemma:R_contr_map}), and the Banach fixed-point theorem (Theorem \ref{thm:Banach}). 
\end{proof}

\begin{remark*}
	The speed of convergence of the fixed-point iteration method may be described by a Lipschitz constant for $R$, denoted here by $q$: $d(\theta,\theta_{n+1}) \leq \frac{q}{1-q} d(\theta_{n+1},\theta_n)$. If we conjecture that $R(x)$ is convex on $[0,\hat\theta]$ then we easily get ($R(x)$ is increasing) that the best Lipschitz constant for $R$ is $q = \sup_{x \in [0,\hat\theta]} R'(x) = R'(\hat\theta) = 1 + \frac{1-\gamma}{\gamma} \ln(1-\gamma)$. However, proving the convexity of $R(x)$ seems very demanding.
\end{remark*}

In the next section we present the results of applying the fixed-point iteration method to estimate $\theta$ for the exemplary random hypergraph following our model.

\section{Experimental results} \label{sec:experiments}

In order to verify the obtained results and the legitimacy of our assumptions, we ran numerous simulations of the model trying different sets of parameters. In this section we present the results of simulated $\tilde{H} = H(H_0, p_v = 0.3, p_e = 0.5, p_d = 0.2, Y_t)$, where the distribution of $Y_t$ was obtained experimentally from a real collaboration network $G$. 
%(we observed similar results by other parameters, check \cite{hyper_desact_arxiv}). We also compare them with the behavior of a real collaboration network $G$. 
$G$ was built upon data extracted from Scopus~\cite{scopus}, these were 239,414 computer science articles published between 1990 and 2018 by 258,145 different authors. Each author was treated as a node and every publication corresponded to a hyperedge between its co-authors. %a group of coauthors of a single article as a hyperedge.

\begin{figure}[!ht]
	%\captionsetup[subfigure]{justification=centering}
	%\begin{subfigure}{.49\textwidth}
		%\vspace{-11pt}
		\centering
		\includegraphics[width=.67\linewidth]{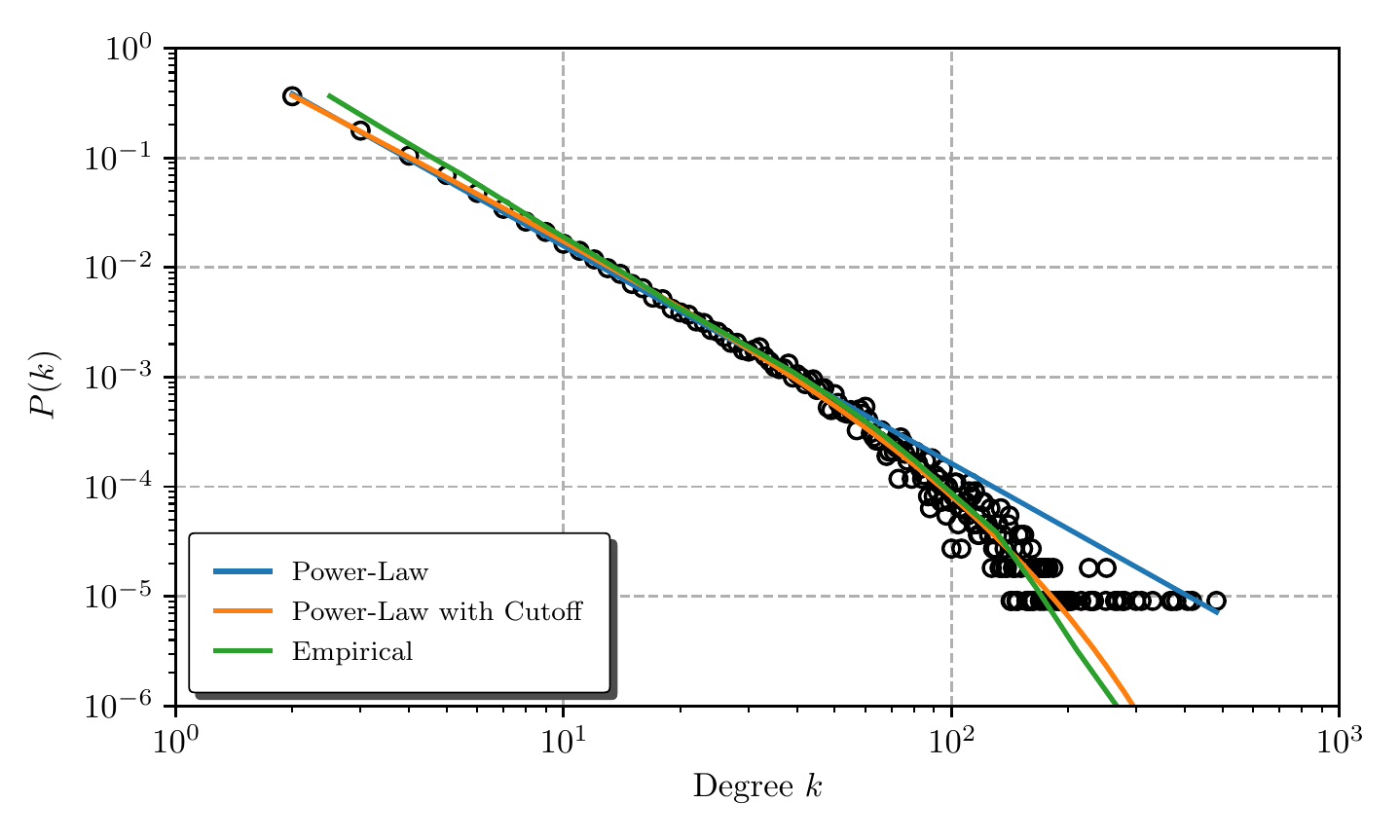}
		\caption{The degree distribution of the Scopus network $G$.} 
		\label{fig:degree_real}
	%\end{subfigure}
	%\hfill
\end{figure}
%\vspace{-10pt}
\begin{figure}[!ht]%{.49\textwidth}
%\begin{wrapfigure}{r}{.3\textwidth}
		%\vspace{-15pt}
		\centering 
		\includegraphics[width=.33\linewidth]{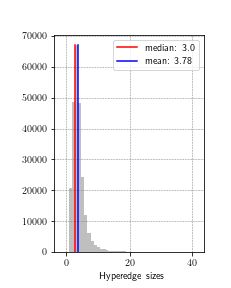}
		\caption{The distribution of sizes of hyperedges in the Scopus network $G$.}
		\vspace{-8pt}
		\label{fig:degree_simul}
		%\vspace{-.4cm}
\end{figure}
%	\caption{The degree distributions of a real-life and simulated hypergraphs.}
%\end{figure}

We used statistical tools from \cite{ClShNe2009} to fit and compare theoretical distributions with the real degree distribution of $G$. One finds the result in Figure \ref{fig:degree_real} which shows that a power-law with an exponential cutoff is a good fit here (this is just one of many examples of real-life networks that follow this distribution \cite{BrCl2019}). Figure \ref{fig:degree_simul} shows the distribution of sizes of hyperedges in $G$ - the one chosen for $Y_t$ in our experiment. %- the simulated one closely corresponds to the theoretical one that we obtain by Theorem \ref{thm:deg_dist}. 

%\begin{figure}[!ht]
%	\captionsetup[subfigure]{justification=centering}
	\begin{figure}[!ht]%{.49\textwidth}
		\vspace{-6pt}
		\centering
		\includegraphics[width=.67\linewidth]{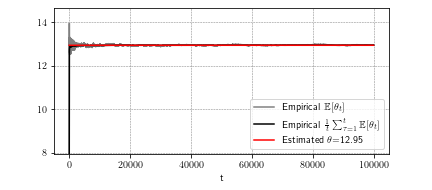}
		\caption{The empirical average deactivated degree (grey and black curves) as compared to the estimated $\theta$ (red line) in $\tilde{H}$.} 
		\label{fig:theta}
	\end{figure}
	%\hfill
	\begin{figure}[!ht]%{.49\textwidth}
		\vspace{-11pt}
		\centering 
		\includegraphics[width=.67\linewidth]{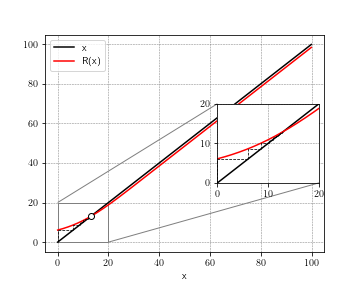}
		\caption{The visualization of the fixed-point iteration method applied to the model $\tilde{H}$, %$H(H_0, p_v = 0.3, p_e = 0.49, p_d = 0.21, Y_t = 3)$,
			starting from $\theta_0 = 0$.}
		\label{fig:theta_iter}
		%\vspace{-.4cm}
	\end{figure}
%	\caption{Experimental results on the parameter $\theta$ and the fixed-point iteration method.}
%\end{figure}

The evolution of the average degree of a vertex selected for deactivation in $\tilde{H}$ compared with the value of $\theta$ calculated using the fixed-point iteration method (Corollary \ref{cor:theta_est}) is presented in Figure \ref{fig:theta}. It shows the convergence of the empirical average degree of a deactivated vertex to the estimated value of $\theta$ which supports both, our Assumption (4) as well as the method for evaluating $\theta$ (see Figure \ref{fig:theta_iter} for its visualization).

%\begin{figure}[!ht]
	%\captionsetup[subfigure]{justification=centering}
	\begin{figure}[!ht]%{.49\textwidth}
		\vspace{-11pt}
		\centering
		\includegraphics[width=.66\linewidth]{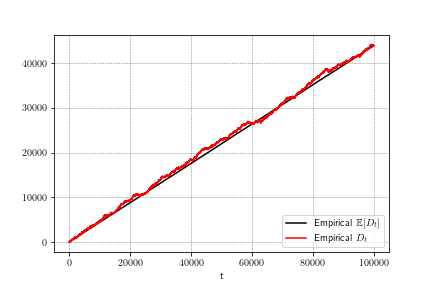}
		\caption{The empirical $\E[D_t]$ and the trajectory of $D_t$ for $\tilde{H}$.} 
		\label{fig:Dt_EDt}
	\end{figure}
	%\hfill
	\begin{figure}[!ht]%{.49\textwidth}
		\vspace{-11pt}
		\centering 
		\includegraphics[width=.66\linewidth]{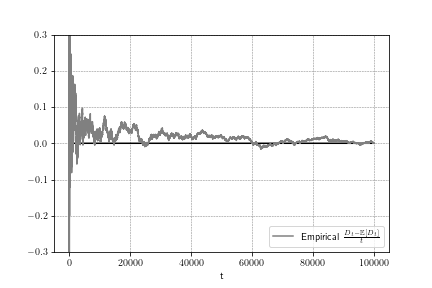}
		\caption{The concentration of $D_t$ for the~hypergraph $\tilde{H}$.}
		\label{fig:Dt_concentration}
		%\vspace{-.4cm}
	\end{figure}
	%\caption{Experimental results on the sum of degrees of active vertices $D_t$ in the simulated hypergraph $\tilde{H}$.}
%\end{figure}

Furthermore, we checked empirically the value of $\E[D_t]$ in the simulated $\tilde{H}$ (we ran 1000 simulations up to 100,000 steps). The empirical $\E[D_t]$ appeared to be linear with the slope $\hat\alpha = 0.438948$ (Figure \ref{fig:Dt_EDt}). We then calculated the slope of the theoretical $\E[D_t]$ using the fixed-point iteration method to compute $\theta$ and then plugging it into equation from Lemma~\ref{fact:Dt/t}. It yielded $\alpha = 0.438184$ which closely corresponds to $\hat\alpha$. %actual $\E[D_t]$ by usingthe fixed-point iteration method to compute $\theta$ and then plugging it into equation from Fact \ref{fact:Dt/t}, which yielded $\alpha = 0.32904168$.
Finally, the result seen in Figure~\ref{fig:Dt_concentration} supports our Assumption~(3) about the concentration of $D_t$.

For the results of simulations conducted with different sets of parameters (e.g. for $Y_t$ following some theoretical distribution, like Poisson), check the Appendix C. % ($D_t = \E[D_t] + o(t)$ whp).

%\begin{figure}
%	\centering
%	\includegraphics[width=0.8\linewidth]{theta-average-0.3-0.49-0.21-3-eps-converted-to.pdf}
%	\caption{The evolution of the empirical average deactivated degree (grey and black curves) as compared to the estimated $\theta$ (red line) in a simulated hypergraph $\tilde{H}$.} %a simulation of the model $H(H_0, p_v = 0.3, p_e = 0.49, p_d = 0.21, Y_t = 3)$.}
%	\label{fig:theta}
%\end{figure}

\section{Conclusions} \label{sec:conclude}

To the best of our knowledge, we have presented the first complex network model which allows for multiary relations and deactivation of elements\footnote{In our model the deactivated elements do not establish new connections any more but they do not disappear from the system and still contribute to the total sum of degrees. We find this setting very useful for real-life applications.}. Both those eventualities occur naturally in real-life systems. We thus believe that the model will find a wide range of applications in many research domains. We have also proved that its degree distribution follows a power-law with an exponential cutoff, which, according to the broad study of Broido and Clauset \cite{BrCl2019}, is the distribution most often observed in nature.

In further research we would like to investigate deeper some natural networks and observe how well our model reflects them. Maybe the need of generalizing the model will occur, e.g., by modifying the form of the attachment function. %or introducing the possibility of adding isolated vertices.
The other interesting direction of study is to make the attachment rule dependent not only on the degrees of vertices but also on their additional own characteristic (called \emph{fitness} in the literature \cite{fitness}). %by preserving the current character of the degree distribution.

%%%%%%%%%%%%%%%%%%%%%%%%%%%%%%%%%%%%%%%%

\bibliographystyle{plainurl}
\bibliography{hyper_bib}

%%%%%%%%%%%%%%%%%%%%%%%%%%%%%%%%%%%%%%%%
\newpage
\appendix

%\section{Appendix}

\section{Preferential attachment hypergraph with vertex deactivation}

\begin{lemma}[Chernoff Bounds, \cite{MU_book}, Chapter 4.2] \label{lemma:Chernoff}
	Let $Z_1,Z_2,\ldots,Z_t$ be independent indicator random variables with $\Pr[Z_i=1]=p_i$ and $\Pr[Z_i=0]=1-p_i$. Let $Z = \sum_{i=1}^{t} Z_i$ and $m = \E[Z] = \sum_{i=1}^{t} p_i$. Then
	\[
	\Pr[|Z-m| \geqslant \delta m] \leq 2 e^{-m \delta^2/3}
	\]
	for all delta $\delta \in (0,1)$.
\end{lemma}

%\begin{lemma}[Chernoff bounds, \cite{MU_book}, Chapter 4.2] \label{lemma:Chernoff}
%	Let $Z_1,Z_2,\ldots,Z_t$ be independent indicator random variables with $\Pr[Z_i=1]=p_i$ and $\Pr[Z_i=0]=1-p_i$. Let $\delta>0$, $Z = \sum_{i=1}^{t} Z_i$ and $\mu = \E[Z] = \sum_{i=1}^{t} p_i$. Then
%	\[
%	\Pr[|Z-\mu| \geqslant \delta\mu] \leq 2 e^{-\mu\delta^2/3}.
%	\]
%\end{lemma}

\begin{corollary} \label{cor:Vt_concentr}
	Since $|V_t|$ follows a binomial distribution with parameters $t$ and $p_v$, setting $\delta = \sqrt{\frac{9 \ln{t}}{p_v t}}$ in Chernoff bounds (Lemma \ref{lemma:Chernoff}) %by Chernoff bounds (consult e.g. Chapter 4.2 in \cite{MU_book}) we get
	we get
	\[
	\Pr[||V_t|-\E[|V_t|]| \geqslant \sqrt{9 p_v t \ln{t}}] \leq 2/t^3.
	\]
\end{corollary}

Below we restate and prove Lemmas~\ref{lemma:Nkt_limit}, \ref{lemma:Y2/D2}, \ref{lemma:A/D}, and Theorem~\ref{thm:deg_dist}.
%\ref{lemma:k2Akt/Dt}, and \ref{lemma:r_formula}.

\begin{repeat_lemma_3} %[See Lemma 4 in \cite{hyper_mod}]
	If $\lim_{t \rightarrow \infty} \frac{\E[N_{k,t}]}{t} \sim c \cdot k^{-\beta} \gamma^k \left(\frac{1}{k} + \delta \right)$ for some positive constants $c, \beta, \gamma, \delta$ then
	\[
	\lim_{t \rightarrow \infty} \E\left[\frac{N_{k,t}}{|V_t|}\right] \sim \frac{c}{p_v} k^{-\beta} \gamma^k \left(\frac{1}{k} + \delta \right).
	\]
	\textnormal{(}Here ``$\sim$'' refers to the limit by $k \rightarrow \infty$.\textnormal{)}
\end{repeat_lemma_3}

\begin{proof}
	Let $B$ denote the event $[||V_t| - \E|V_t|| < \sqrt{9 p_vt\ln{t}}]$ and $B^C$ its complement. We have
	\[
		\E\left[\frac{N_{k,t}}{|V_t|}\right] = \E\left[\frac{N_{k,t}}{|V_t|}|B\right] \Pr[B] + \E\left[\frac{N_{k,t}}{|V_t|}|B^C\right] \Pr[B^C].
	\]
	Since $N_{k,t} \leq |V_t|$ and $\E|V_t| = p_v t$, by Corollary \ref{cor:Vt_concentr} we get
	\[
		\E\left[\frac{N_{k,t}}{|V_t|}\right] \leq \frac{\E[N_{k,t}]}{\E|V_t|-\sqrt{9 p_vt\ln{t}}} \cdot 1+ 1 \cdot \frac{2}{t^3} \sim \frac{\E[N_{k,t}]}{p_v t},
	\]
	and on the other hand
	\[
		\E\left[\frac{N_{k,t}}{|V_t|}\right] \geq \E\left[\frac{N_{k,t}}{|V_t|}|B\right] \Pr[B] \geq \frac{\E[N_{k,t}]}{\E|V_t|+\sqrt{9 p_vt\ln{t}}} \left(1-\frac{2}{t^3}\right) \sim \frac{\E[N_{k,t}]}{p_v t}.
	\]
\end{proof}

\begin{repeat_lemma_4}
	Assume that $\E[Y_t] = \mu$ for all $t>0$ and $\V[Y_t] = o(t)$. Then %and $D_t = \E[D_t] + o(t)$ whp. Then
		\[
		\E\left[\frac{Y_t^2}{D_{t-1}^2}\right] = o\left(\frac{1}{t}\right).
		\]
\end{repeat_lemma_4}
\begin{proof}
	By the fact that $D_{t-1} \geq A_{t-1}$ and $Y_t$ is independent of $A_{t-1}$ we have
	\begin{equation} \label{eq:Y/D}
		\E\left[\frac{Y_t^2}{D_{t-1}^2}\right] \leq \E\left[\frac{1}{A_{t-1}^2}\right] \E[Y_t^2] = \E\left[\frac{1}{A_{t-1}^2}\right] (\V[Y_t] + \mu^2).
	\end{equation}
	Note that $A_{t-1}$ follows a binomial distribution with parameters $t-1$ and $p = (p_v-p_d)$ (recall that we assume  $p_v>p_d$ throughout the paper) thus setting $\delta = \sqrt{\frac{6 \ln{t}}{p(t-1)}}$ in Chernoff bounds (Lemma \ref{lemma:Chernoff}) we may write 
	\[
	\Pr[|A_{t-1} - (t-1) p| \geq \sqrt{6 p \cdot (t-1) \ln{t} }] \leq 2/t^2.
	\]
	Let $B$ denote the event $[|A_{t-1} - (t-1) p| < \sqrt{p \cdot t \ln{t^2}}]$ and $B^C$ its complement. Note that $A_{t-1} \geq 1$. We have
	\begin{equation} \label{eq:1/A2}
		\begin{split}
			\E\left[\frac{1}{A_{t-1}^2}\right] & = \E\left[\frac{1}{A_{t-1}^2}| B\right] \Pr[B] + \E\left[\frac{1}{A_{t-1}^2}| B^C \right] \Pr[B^C] \\
			& \leq \frac{1}{\left(p(t-1) - \sqrt{p \cdot t \ln{t^2}}\right)^2} \cdot 1 + 1 \cdot \frac{2}{t^2} = \Theta\left(\frac{1}{t^2}\right).
		\end{split}
	\end{equation}
	Thus by (\ref{eq:Y/D}) and (\ref{eq:1/A2}), since $\V[Y_t] = o(t)$, we obtain
	\[
	\E\left[\frac{Y_t^2}{D_{t-1}^2}\right] \leq \frac{1}{t^2} (\V[Y_t] + \mu^2) = o\left(\frac{1}{t}\right).
	\]
\end{proof}

\begin{repeat_lemma_5}
		Assume that $D_t = \E[D_t] + o(t)$ whp. Then for each $k \geq 1$
	\[
	\E\left[\frac{A_{k,t}}{D_t}\right] = \frac{\E[A_{k,t}]}{\E[D_t]} + o(1).
	\]
\end{repeat_lemma_5}
\begin{proof}
	Denote the event $[D_t = \E[D_t] + o(t)]$ by $B$ and its complement by $B^C$. Since $\Pr[B^C] = o(1)$, $\frac{A_{k,t}}{D_t} \leq 1$ and $\E[D_t] = \Omega(t)$ (note that $\E[D_t] \geq \E[A_t] = 1 +(p_v-p_d)t$ and we assume $p_v>p_d$), we have
	\[
	\begin{split}
		\E\left[\frac{A_{k,t}}{D_t}\right] & = \E\left[\frac{A_{k,t}}{D_t}|B\right]\Pr[B] + \E\left[\frac{A_{k,t}}{D_t}|B^C\right]\Pr[B^C]\\ 
		& = \frac{\E[A_{k,t}]}{\E[D_t]+o(t)} + o(1) = \frac{\E[A_{k,t}]}{\E[D_t]} + o(1).
	\end{split}
	\]
\end{proof}

\begin{repeat_theorem_6}
	Consider a hypergraph $H = H(H_0,p_e,p_v,Y)$ for any $t>0$. By Assumptions (1-4) the degree distribution of $H$ follows a power-law with an exponential cutoff, i.e.,
	\[
	\E\left[\frac{N_{k,t}}{|V_t|}\right] \sim c \cdot k^{-\beta} \gamma^k \left(\frac{1}{k} + \delta \right)
	\]
	%for
	\[
	\begin{split}
		for \quad \beta & = \frac{\mu(p_v+p_e)-p_d \theta}{p_v(\mu-1)+p_e \mu + p_d}, \quad \quad
		\gamma  = \frac{p_v(\mu-1)+p_e \mu}{p_v(\mu-1)+p_e\mu+p_d}, \\
		\delta & = \frac{p_d}{\mu(p_v+p_e)-p_d \theta}, \quad \quad \quad \quad c  = \frac{\beta \cdot \Gamma(1+\beta)}{\gamma},
	\end{split}
	\]
	where $\Gamma(x)$ stands for the gamma function \textnormal{(}$\Gamma(x) = \int_{0}^{\infty} t^{x-1}e^{-t} \,\mathrm{d}t$\textnormal{)}.
\end{repeat_theorem_6}

\begin{proof}
	We take a standard master equation approach that can be found e.g. in Chung and Lu book \cite{ChLu_book} about complex networks. %or Avin et al. paper \cite{ALP_hyper} on preferential attachment hypergraphs.
	However, we apply it separately to the number of active vertices and the number of deactivated vertices.
	
	Recall that $N_{k,t}$ denotes the number of vertices of degree $k$ at time $t$. We need to show that $\lim_{t \rightarrow \infty} \E\left[\frac{N_{k,t}}{|V_t|}\right] \sim c \cdot k^{-\beta} \gamma^k \left(\frac{1}{k} + \delta \right)$ 
	%	\[
	%	\lim_{t \rightarrow \infty} \E\left[\frac{N_{k,t}}{|V_t|}\right] \sim c \cdot k^{-\beta} \gamma^k \left(\frac{1}{k} + \delta \right)
	%	\]
	for the proper constants $c, \beta, \gamma$ and $\delta$.
	However, by Lemma \ref{lemma:Nkt_limit} we know that it suffices to show that
	\[
	\lim_{t \rightarrow \infty} \frac{\E[N_{k,t}]}{t} \sim p_v \cdot c \cdot k^{-\beta} \gamma^k \left(\frac{1}{k} + \delta \right).
	\]
	
	Recall that $N_{k,t} = A_{k,t} + I_{k,t}$. First, let us evaluate $\lim_{t \rightarrow \infty} \frac{\E[A_{k,t}]}{t}$ using the mathematical induction on $k$. In this part we follow closely the lines of the proof that can be found in~\cite{ALP_hyper}. Consider the case $k=1$. Since $H_0$ consists of a single hyperedge of cardinality 1 over
	a single vertex, we have $A_{1,0} = 1$. To formulate a master equation, let us make the following observation for $t \geq 1$. An active vertex remains in $\mathcal{A}_{1,t}$ if it had degree 1 at step $t - 1$ and was neither selected to a hyperedge, nor deactivated. Recall that a vertex from $\mathcal{A}_{t-1}$ is chosen at step $t$ in a single trial to the new hyperedge with probability $1/D_{t-1}$ thus the chance that it won't be selected to the hyperedge of cardinality $y$ equals $(1-1/D_{t-1})^y$.
	Also, in each step, with probability $p_v$, a single new active vertex of degree $1$ is added to the hypergraph.
	Let $\mathcal{F}_t$ denote a $\sigma$-algebra associated with the probability space at step $t$. 
	For $t \geq 1$ we have
	\begin{equation} \label{equation:E[A_(1,t)]-recurrence_2}
		\begin{split}
			\E[A_{1,t} | \mathcal{F}_{t-1}] 
			& = p_v A_{1,t-1} \left(1 - \frac{1}{D_{t-1}}\right)^{Y_t - 1}
			+ p_e A_{1,t-1} \left(1 - \frac{1}{D_{t-1}}\right)^{Y_t}     \\
			& \quad 
			+ p_d A_{1,t-1} \left(1 - \frac{1}{D_{t-1}}\right)
			+ p_v. 
		\end{split}
	\end{equation}
	After taking the expectation on both sides of (\ref{equation:E[A_(1,t)]-recurrence_2}) we derive upper and lower bounds on $\E[A_{1,t}]$. By Bernoulli's inequality ($(1+x)^n \geq 1+nx$ for $n \in \mathbb{N}$ and $x\geq-1$), Lemma \ref{lemma:A/D} (thus by Assumption (3) necessary for it) and the independence of $Y_t$ from $A_{1,t-1}$ and $D_{t-1}$ we obtain
	\begin{equation} \label{equation:E[A_(1,t)]-lower-bound_2}
		\begin{split} 
			\E[A_{1,t}]
			&\geq p_v \E\left[A_{1,t-1} \left({1 - \frac{Y_t - 1}{D_{t-1}}}\right)\right]
			+ p_e \E\left[A_{1,t-1} \left({1 - \frac{Y_t}{D_{t-1}}}\right)\right] \\
			&\quad
			+ p_d \E\left[{A_{1,t-1} \left({1 - \frac{1}{D_{t-1}}}\right)}\right]
			+ p_v \\
			&= p_v \E[{A_{1,t-1}}] \left({1 - \frac{\E[{Y_t}] - 1}{\E[{D_{t-1}}]}}\right)
			+ p_e \E[{A_{1,t-1}}] \left({1 - \frac{\E[Y_t]    }{\E[{D_{t-1}}]}}\right) \\
			&\quad
			+ p_d \E[{A_{1,t-1}}] \left({1 - \frac{1}{\E[{D_{t-1}}]}}\right)
			+ p_v - o(1) \\
			&= \E[A_{1,t-1}] \left({
				1 - \frac{p_v (\mu - 1) + p_e \mu + p_d}{\E[{D_{t-1}}]}}\right) + p_v - o(1). 
		\end{split}
	\end{equation}
	
	On the other hand, since $(1 - x)^n \leq 1/(1 + nx)$ for $x\in[0, 1]$ and $n\in\mathbb{N}$, and $A_{1,t-1} \leq t$, by Lemma \ref{lemma:Y2/D2} (thus by Assumptions (1) and (2) necessary for it) and Lemma \ref{lemma:A/D} (thus by Assumption (3)) we have
	\begin{equation} \label{equation:E[A_(1,t)]-upper-bound_2}
		\begin{split}
			\E[{A_{1,t}}]
			&\leq p_v \E\left[{\frac{A_{1,t-1}}{1 + (Y_t - 1)/D_{t-1}}}\right]
			+ p_e \E\left[{\frac{A_{1,t-1}}{1 +  Y_t     /D_{t-1}}}\right] \\
			&\quad
			+ p_d\E\left[{A_{1,t-1} \left({1 - \frac{1}{D_{t-1}}}\right)}\right]
			+ p_v \\
			&=   p_v \E\left[{A_{1,t-1} \left({1 - \frac{Y_t - 1}{D_{t-1} + Y_t - 1}}\right)}\right]
			+ p_e \E\left[{A_{1,t-1} \left({1 - \frac{Y_t    }{D_{t-1} + Y_t}}\right)}\right] \\
			&\quad
			+ p_d \E\left[{A_{1,t-1} \left({1 - \frac{1}{D_{t-1}}}\right)}\right]
			+ p_v \\
			&\leq p_v \E\left[{A_{1,t-1} \left({1 - \frac{Y_t - 1}{D_{t-1}} + \frac{(Y_t - 1)^2}{D_{t-1}^2}}\right)}\right]
			+ p_e \E\left[{A_{1,t-1} \left({1 - \frac{Y_t    }{D_{t-1}} + \frac{ Y_t     ^2}{D_{t-1}^2}}\right)}\right] \\
			&\quad
			+ p_d \E\left[{A_{1,t-1} \left({1 - \frac{1}{D_{t-1}}}\right)}\right]
			+ p_v \\
			&= p_v \E[{A_{1,t-1}}] \left({1 - \frac{\E[{Y_t}] - 1}{\E[{D_{t-1}}]}}\right)
			+ p_e \E[{A_{1,t-1}}] \left({1 - \frac{\E[{Y_t}]    }{\E[{D_{t-1}}]}}\right) \\
			&\quad
			+ p_d \E[{A_{1,t-1}}] \left({1 - \frac{1}{\E[{D_{t-1}}]}}\right)
			+ p_v 
			+ p_v \E\left[{O(t) \frac{(Y_t - 1)^2}{D_{t-1}^2}}\right]
			+ p_e \E\left[{O(t) \frac{Y_t^2}{D_{t-1}^2}}\right] + o(1)\\
			&= \E[{A_{1,t-1}}] \left({
				1 - \frac{p_v (\mu - 1) + p_e \mu + p_d}{\E[{D_{t-1}}]}
			}\right) + p_v + o(1). 
		\end{split}
	\end{equation}
	From (\ref{equation:E[A_(1,t)]-lower-bound_2}) and (\ref{equation:E[A_(1,t)]-upper-bound_2}) we get
	\[
	\E[{A_{1,t}}] = \E[{A_{1,t-1}}] \left({
		1 - \frac{p_v (\mu - 1) + p_e \mu + p_d}{\E[{D_{t-1}}]}
	}\right) + p_v + o(1). 
	\]
	Now, we apply Lemma \ref{lemma:rec_seq} to the above equation choosing
	\[
	a_t  = \E[{A_{1,t}}], \quad
	b_t  =  \frac{p_v (\mu - 1) + p_e \mu + p_d}{\E[{D_{t-1}}]/t}, \quad
	c_t  =  p_v + o(1).
	\]
	We have $\lim_{t\to\infty} c_t = p_v$ and, by Fact \ref{fact:Dt/t} (thus by Assumptions (1) and (4) implying it), $\lim_{t\to\infty} b_t = \frac{p_v (\mu-1) + p_e \mu + p_d}
	{\mu (p_v + p_e) - p_d \theta} =: 1/\beta$ thus 
	\[
	\lim_{t\to\infty} \frac{a_t}{t} = \lim_{t\to\infty} \frac{\E[{A_{1,t}}]}{t} = \frac{p_v}{1 + 1/\beta} =: \bar{A}_1.
	\]
	Now, we assume that the limit $\lim_{t\to\infty} \frac{\E[{A_{k-1,t}}]}{t}$ exists and equals $\bar{A}_{k-1}$ and we will show by induction on $k$ that the analogous limit for $\E[{A_{k,t}}]$ exists. Let us again formulate a master equation, this time for $k>1$. We have $A_{k,0}=0$ and for $t \geq 1$ an active vertex appears in $A_{k,t}$ if it was active at step $t-1$, had degree $k-l$ and was chosen exactly $l$ times to a hyperedge, or it had degree $k$ and was not selected for deactivation. Let $B({l,n,p}) = \binom{n}{l} p^l (1-p)^{n-l}$. We have
	\[
	\begin{split}
		\E[{A_{k,t} | \mathcal{F}_{t-1}} ]
		&= p_v \smashoperator{\sum_{l=0}^{\min\{Y_t-1, k-1\}}} 
		A_{k-l,t-1} B\left({l, Y_t - 1, \frac{k-l}{D_{t-1}}}\right)
		+ p_e \smashoperator{\sum_{l=0}^{\min\{Y_t, k-1\}}} 
		A_{k-l,t-1} B\left({l, Y_t , \frac{k-l}{D_{t-1}}}\right) \\
		&\quad
		+ p_d A_{k,t-1} \left({1 - \frac{k}{D_{t-1}}}\right).
	\end{split}
	\]
	%We will now extract two first summands from each sum and group them together with the last term of the above expression since (as we prove below) they compose the significant term.
	Taking the expectation on both sides we get
	\[
	\E[{A_{k,t}}] = 
	\E[\psi] + p_v \E[{\varphi(Y_t - 1)}] + p_e \E[{\varphi(Y_t)}], \quad \textnormal{where}
	\]
	%where
	\[
	\begin{split}
		\psi 
		&= p_v \sum_{l=0}^1
		A_{k-l,t-1} B\left({l, Y_t - 1, \frac{k-l}{D_{t-1}}}\right)
		+ p_e \sum_{l=0}^1 
		A_{k-l,t-1} B\left({l, Y_t, \frac{k-l}{D_{t-1}}} \right) \\
		&\quad + p_d A_{k,t-1} \left({1 - \frac{k}{D_{t-1}}}\right) \\
		&= A_{k,t-1} \left({
			p_v \left({1 - \frac{k}{D_{t-1}}}\right)^{Y_t-1} + 
			p_e \left({1 - \frac{k}{D_{t-1}}}\right)^{Y_t  } +
			p_d \left({1 - \frac{k}{D_{t-1}}}\right)
		}\right) \\
		&\quad
		+ A_{k-1,t-1} \frac{k-1}{D_{t-1}} \left({
			p_v (Y_t - 1) \left({1 - \frac{k-1}{D_{t-1}}}\right)^{Y_t-2} + 
			p_e  Y_t     \left( {1 - \frac{k-1}{D_{t-1}}}\right)^{Y_t-1}
		}\right)
	\end{split}
	\]
	and $\quad \varphi(n) = \smashoperator{\sum_{l=2}^{\min\{n, k-1\}}} A_{k-l,t-1} B\left({l, n, \frac{k-l}{D_{t-1}}}\right)$.
	%	\begin{equation}
		%	\begin{split}
			%		\varphi(n) = \smashoperator{\sum_{l=2}^{\min\{n, k-1\}}} A_{k-l,t-1} B\left({l, n, \frac{k-l}{D_{t-1}}}\right).
			%	\end{split}
		%	\end{equation}
	
	\noindent
	We will show that only the term $\E[\psi]$ is significant and that the terms $\E[{\varphi(Y_t - 1)}]$ and $\E[{\varphi(Y_t)}]$ converge to 0 as $t\to\infty$. We have
	\[
	\begin{split}
		\varphi(Y_t) 
		&\leq \sum_{l=2}^{k-1} A_{k-l,t-1} 
		\binom{Y_t}{l} \left({\frac{k - l}{D_{t-1}}}\right)^l \left({1 - \frac{k-l}{D_{t-1}}}\right)^{Y_t - l} \\
		&\leq O(t) \sum_{l=2}^{k-1} 
		\binom{Y_t}{l} \left({\frac{k - l}{D_{t-1}}}\right)^l \left({1 - \frac{k-l}{D_{t-1}}}\right)^{Y_t - l} \\
		&\leq O(t) \sum_{l=2}^{k-1} 
		Y_t^l \left({\frac{k}{D_{t-1}}}\right)^l\left( {1 - \frac{1}{D_{t-1}}}\right)^{Y_t - k + 1} \\
		&\leq O(t)
		\frac{Y_t^2 k^2}{D_{t-1}^2} e^{-Y_t/D_{t-1}} e^{k-1} \sum_{l=2}^{k-1} \left({\frac{Y_t k}{D_{t-1}}}\right)^{l-2} \\
		&= O(t)
		\frac{Y_t^2}{D_{t-1}^2} e^{-Y_t/D_{t-1}} \sum_{l=2}^{k-1} \left({\frac{Y_t k}{D_{t-1}}}\right)^{l-2}.
	\end{split}
	\]
	Then, if $Y_t \leq D_{t-1}$, we have 
	\[
	\varphi(Y_t) \leq 
	O(t) \frac{Y_t^2}{D_{t-1}^2} k^{k-2} =
	O(t) \frac{Y_t^2}{D_{t-1}^2}.
	\]
	Otherwise, 
	\[
	\begin{split}
		\varphi(Y_t) 
		&\leq O(t) \frac{Y_t^2}{D_{t-1}^2} e^{-Y_t/D_{t-1}} 
		\frac{(Y_t k / D_{t-1})^{k-2} - 1}{(Y_t k / D_{t-1}) - 1} \\ & \leq O(t) \frac{Y_t^2}{D_{t-1}^2} e^{-Y_t/D_{t-1}} 
		\frac{(Y_t   / D_{t-1})^{k-2}}{k - 1}  k^{k-2} \\
		&\leq O(t) \frac{Y_t^2}{D_{t-1}^2} e^{-(k-2)}
		\frac{(k - 2)^{k-2}}{k - 1}  k^{k-2} = O(t) \frac{Y_t^2}{D_{t-1}^2},
	\end{split}
	\]
	where the last inequality follows from the fact that $e^{-x} x^{\alpha}$ is maximized at $x = \alpha$. Hence by Lemma \ref{lemma:Y2/D2} (thus by Assumptions (1) and (2)) in both above cases we get $\E[\varphi(Y_t)] = o(1)$ and, similarly, $\E[\varphi(Y_t-1)] = o(1)$. Now, we derive the bounds for $\E[\psi]$ analogous to the ones derived for $\E[A_{1,t}]$.
	\begin{equation} \label{equation:psi-lower-bound}
		\begin{split}
			\E[{\psi}]
			&= \E\left[{
				A_{k,t-1} \left({
					p_v \left({1 - \frac{k}{D_{t-1}}}\right)^{Y_t-1} + 
					p_e \left({1 - \frac{k}{D_{t-1}}}\right)^{Y_t  } +
					p_d \left({1 - \frac{k}{D_{t-1}}}\right)
				}\right)
			}\right] \\
			& \quad + \E\left[{
				A_{k-1,t-1} \frac{k-1}{D_{t-1}} \left({
					p_v (Y_t - 1) \left({1 - \frac{k-1}{D_{t-1}}}\right)^{Y_t-2} + 
					p_e  Y_t     \left( {1 - \frac{k-1}{D_{t-1}}}\right)^{Y_t-1}
				}\right)
			}\right] \\
			&\geq \E\left[{
				A_{k,t-1} \left({
					p_v \left({1 - \frac{(Y_t-1) k}{D_{t-1}}}\right) + 
					p_e \left({1 - \frac{ Y_t    k}{D_{t-1}}}\right) +
					p_d \left({1 - \frac{k}{D_{t-1}}}\right)
				}\right)
			}\right] \\
			&\quad +\E\left[ {
				A_{k-1,t-1} \frac{k-1}{D_{t-1}} \left({
					1 - \frac{(k-1)(Y_t-1)}{D_{t-1}}
				}\right) \left({
					p_v (Y_t - 1) + p_e Y_t
				}\right)
			}\right] \\
			&= \E[{A_{k,t-1}}] \left({
				1 - \frac{k \left({p_v (\mu-1) + p_e \mu + p_d}\right)}{\E[{D_{t-1}}]}
			}\right) \\
			&\quad
			+ \E[{A_{k-1,t-1}}] \frac{(k-1) \left({p_v (\mu-1) + p_e \mu}\right)}{\E[{D_{t-1}}]} + o(1).
		\end{split}
	\end{equation}
	On the other hand,
	\begin{equation} \label{equation:psi-upper-bound}
		\begin{split}
			\E[{\psi} ]
			&\leq \E\left[{A_{k,t-1}} \left({
				1 - \frac{k (p_v (Y_t-1) + p_e Y_t + p_d)}{D_{t-1}}
			}\right) \right] + \E\left[{O(t) \frac{Y_t^2}{D_{t-1}^2}}\right] \\
			&\quad
			+ \E\left[{A_{k-1,t-1} \frac{k-1}{D_{t-1}} (p_v (Y_t - 1) + p_e Y_t)}\right] \\
			&= \E[{A_{k,t-1}}] \left({
				1 - \frac{k \left({p_v (\mu-1) + p_e \mu + p_d}\right)}{\E[{D_{t-1}}]}
			} \right)\\
			&\quad
			+ \E[{A_{k-1,t-1}}] \frac{(k-1) \left({p_v (\mu-1) + p_e \mu}\right)}{\E[{D_{t-1}}]} + o(1).
		\end{split}
	\end{equation}
	By (\ref{equation:psi-lower-bound}) and (\ref{equation:psi-upper-bound}) we get
	\begin{equation}
		\begin{split}
			\E[{A_{k,t}}]
			&= \E[{A_{k,t-1}}] \left({
				1 - \frac{k \left({p_v (\mu-1) + p_e \mu + p_d}\right)}{\E[{D_{t-1}}]}
			}\right) \\
			&\quad
			+ \E[{A_{k-1,t-1}}] \frac{(k-1) \left({p_v (\mu-1) + p_e \mu}\right)}{\E[{D_{t-1}}]} + o(1).
		\end{split}
	\end{equation}
	Recall that by the induction assumption $\lim_{t\to\infty} \E[{A_{k-1,t}}]/t = \bar{A}_{k-1}$. Now, we apply again Lemma \ref{lemma:rec_seq} to the above equation choosing
	\[
	\begin{split}
		a_t &= \E[{A_{k,t}}], \quad
		b_t = \frac{
			k (p_v (\mu-1) + p_e \mu + p_d)
		}{
			\E[{D_{t-1}}]/t
		}, \\
		c_t &= \frac{\E[{A_{k-1,t-1}}]}{t}
		\frac{
			(k - 1) (p_v (\mu-1) + p_e \mu )
		}{
			\E[{D_{t-1}}]/t
		}
		+ o(1).
	\end{split}
	\]
	By Fact \ref{fact:Dt/t} (thus by Assumptions (1) and (4)) we have 
	\[
	\lim_{t\to\infty} b_t = k / \beta \quad \textnormal{and} \quad 
	\lim_{t\to\infty} c_t = \bar{A}_{k-1}
	\frac{
		(k - 1) (p_v (\mu-1) + p_e \mu)
	}{
		\mu (p_v + p_e) - p_d \theta
	}
	\]
	thus
	\begin{equation}
		\lim_{t\to\infty} \frac{a_t}{t}
		= \lim_{t\to\infty} \frac{\E[{A_{k,t}}]}{t}
		= \bar{A}_k
		= \bar{A}_{k-1} \frac{
			(k - 1) \gamma
		}{
			k + \beta
		},
	\end{equation}
	where $\gamma = \frac{p_v (\mu - 1) + p_e \mu}{p_v (\mu - 1) + p_e \mu + p_d}$. 
	%	\[%\label{equation:gamma}
	%		\gamma = \frac{p_v (\mu - 1) + p_e \mu}{p_v (\mu - 1) + p_e \mu + p_d}.
	%	\]
	Thus we got
	\[
	\begin{split}
		\bar{A}_1 &=  p_v \beta \frac{1}{(1 + \beta)}, \quad 
		\bar{A}_2 = p_v \beta \frac{\gamma}{(1 + \beta)(2 + \beta)}, \quad \ldots,  \\
		\bar{A}_k & = p_v \beta \frac{\gamma^{k-1} (k-1)!}{(1 + \beta)(2 + \beta) \dots (k + \beta)}.
	\end{split}
	\]
	Since $\lim_{k \rightarrow \infty} \frac{\Gamma(k)k^{\alpha}}{\Gamma(k+\alpha)} = 1$ for constant $\alpha \in \mathbb{R}$ we have
	\begin{equation} \label{eq:Ak_1}
		\lim_{t\to\infty} \frac{\E[{A_{k,t}}]}{t} = \bar{A}_k =
		\frac{p_v \beta}{\gamma} 
		\frac{\gamma^k \Gamma(1 + \beta) \Gamma(k)}
		{\Gamma(k + \beta + 1)}
		\sim
		p_v \cdot c \cdot \gamma^k k^{-(\beta + 1)}
	\end{equation}
	with $c = \frac{\beta \cdot \Gamma(1 + \beta)}{\gamma}$.
	%	\[
	%		c = \frac{\beta \cdot \Gamma(1 + \beta)}{\gamma}.% = \frac{\mu(p_v + p_e)-p_d \theta}{p_v(\mu-1)+p_e\mu} \cdot \Gamma(1+ \beta).
	%	\]
	
	Now, let us evaluate $\lim_{t\to\infty} \frac{\E[{I_{k,t}}]}{t}$. We have $I_{k,0} = 0$ for all $k \geq 1$. For $t \geq 1$ the expected number of inactive vertices of degree $k \geq 1$ at step $t$, given $\mathcal{F}_{t-1}$, can be expressed as
	\[
	\E[{I_{k,t} | \mathcal{F}_{t-1}}] =
	I_{k,t-1} + p_d A_{k,t-1} \frac{k}{D_{t-1}},
	\]
	since inactive vertices of degree $k$ remain in $I_{k,t}$ forever and a vertex of degree $k$ becomes inactive if it was selected in step $t-1$ for deactivation.
	Taking the expectation on both sides, by Lemma \ref{lemma:A/D} (thus by Assumption (3)), we obtain
	\[
	\E[{I_{k,t}}] = 
	\E[{I_{k,t-1}} ]
	+ p_d \E[{A_{k,t-1}}] \frac{k}{\E[{D_{t-1}}]} + o(1).
	\]
	Then, by Fact \ref{fact:Dt/t} (thus by Assumptions (1) and (4)),
	\[
	\begin{split}
		\lim_{t\to\infty} \left(\E[I_{k,t}] - \E[{I_{k,t-1}}]\right)
		&= \lim_{t\to\infty} p_d k
		\frac{\E[{A_{k,t-1}}]}{t}
		\frac{t}{\E[{D_{t-1}}]} + o(1) \\
		&= \bar{A}_k \frac{p_d k}{(p_v + p_e) \mu - p_d \theta} = \bar{A}_k k \delta,
	\end{split}
	\]
	where $\delta = \frac{p_d}{(p_v + p_e) \mu - p_d \theta}$.
	%	\[
	%		\delta = \frac{p_d}{(p_v + p_e) \mu - p_d \theta}.
	%	\]
	And, by Stolz--Ces\`aro theorem (Theorem \ref{thm:S-C}), we obtain
	\begin{equation} \label{eq:Ik_2}
		\bar{I}_k := 
		\lim_{t\to\infty} \frac{\E[{I_{k,t}}]}{t} = 
		\lim_{t\to\infty} (\E[{I_{k,t}}] - \E[{I_{k,t-1}}]) = 
		\bar{A}_k k \delta.
	\end{equation}
	Finally, by (\ref{eq:Ak}) and (\ref{eq:Ik_2})
	\[
	\lim_{t\to\infty} \frac{\E[{N_{k,t}}]}{t} = \lim_{t\to\infty} \frac{\E[{A_{k,t}}]+\E[{I_{k,t}}]}{t}  = 
	\bar{A}_k + \bar{I}_k = 
	\bar{A}_k (1 + k\delta) \sim
	p_v \cdot c \cdot k^{-\beta} \gamma^k  \left({\frac{1}{k} + \delta}\right).
	\]
	%    and, by Lemma \ref{lemma:Nkt_limit},
	%	\[
	%	\lim_{t\to\infty} \frac{\E[{N_{k,t}}]}{|V_t|} \sim c \cdot   k^{-\beta}\gamma^k \left({\frac{1}{k} + \delta}\right).
	%	\]
\end{proof}

\newpage
\section{Estimating the limiting value $\theta$}

\begin{repeat_lemma_8}
	Assume that $\Pr[D_t \neq \E[D_t] + o(t)] = o(1/t)$. Then
	\[
		\E\left[ \frac{\sum_{k \geq 1} k^2 A_{k,t}}{D_t}\right] = \frac{\E[\sum_{k \geq 1} k^2 A_{k,t}]}{\E[D_t]} + o(1).
	\]
\end{repeat_lemma_8}
\begin{proof}
	Denote the event $[D_t = \E[D_t] + o(t)]$ by $B$ and its complement by $B^C$. Let us work assuming that $k \leq t$ (indeed, in our model it is very unlikely that a vertex achieves degree greater than $t$ after $t$ steps). Since
	\[
	D_t = \sum_{k \geq 1} k A_{k,t}, \quad \frac{\sum_{k \geq 1} k^2 A_{k,t}}{D_t} \leq \frac{t \sum_{k \geq 1} k A_{k,t}}{D_t} = t, \quad \Pr[B^C] = o(1/t),
	\]
	and $\E[D_t] = \Omega(t)$ (note that $\E[D_t] \geq \E[A_t] = 1 +(p_v-p_d)t$ and we assume $p_v>p_d$), we have
	\[
	\begin{split}
		\E\left[\frac{\sum_{k \geq 1} k^2 A_{k,t}}{D_t}\right] & = \E\left[\frac{\sum_{k \geq 1} k^2 A_{k,t}}{D_t}|B\right]\Pr[B] + \E\left[\frac{\sum_{k \geq 1} k^2 A_{k,t}}{D_t}|B^C\right]\Pr[B^C]\\ 
		& \leq \frac{\E[\sum_{k \geq 1} k^2 A_{k,t}]}{\E[D_t]+o(t)} + t \cdot o(1/t) = \frac{\E[\sum_{k \geq 1} k^2 A_{k,t}]}{\E[D_t]} + o(1).
	\end{split}
	\]
\end{proof}

\begin{repeat_lemma_11}
	The function $R(x) = \frac{F(2,2;\rho(x);\gamma)}{F(1,2;\rho(x);\gamma)}$ can be also expressed as
	\[%\label{equation:R-new}
	R(x)
	= x - \frac{p_v}{p_d} + 
	\frac{1}{1 - \gamma}
	\frac{\rho(x) - 1}{F(1,2;\rho(x);\gamma)},
	\]
	where $\gamma$ and $\rho(x)$ are as in Theorem \ref{thm:R_fixed}.
\end{repeat_lemma_11}

\begin{proof}
	We will use the Gauss' contiguous relations (consult~\cite{AbSt1972}). Let $a,b,c,z \in \mathbb{C}$ with $|z|<1$ and $c \notin \mathbb{Z}_{\leq 0}$. Let $F(z) = F(a, b; c; z)$, $F(a+,z) = F(a+1, b; c; z)$ and $F(a-,z) = F(a-1, b; c; z)$. Then
	\[ %\begin{equation}\label{equation:contiguous-relation}
	a (F(a+,z) - F(z)) =
	\frac
	{(c-a) F(a-,z) + (a-c+bz) F(z)}
	{1-z}
	\]
	%where $F = F(a, b; c; z)$, $F(a+) = F(a+1, b; c; z)$ and $F(a-) = F(a-1, b; c; z)$.
	%	We get %From (\ref{equation:contiguous-relation}), we obtain
	%	\[ %\label{equation:F(a+)}
	%		F(a+,z) = \frac
	%		{\big(2a-c+(b-a)z\big) F(z) + (c-a) F(a-,z)}
	%		{a(1-z)},
	%	\]
	%	and thus
	which is equivalent to
	\begin{equation}\label{equation:F(a+)/F}
		\frac{F(a+,z)}{F(z)} =
		\frac{2a - c + (b-a)z}
		{a(1-z)} +
		\frac{(c-a) F(a-,z)}
		{a(1-z) F(z)}.
	\end{equation}
	Since $F(0, b; c; z) = 1$, plugging $a = 1$, $b = 2$, $c = \rho(x)$ and $z = \gamma$ into (\ref{equation:F(a+)/F}), we get the result. %obtain Equation~\ref{equation:R-new}.
\end{proof}

\begin{repeat_lemma_12}
	$R(x)$ strictly increases on $[0, \hat\theta]$.
\end{repeat_lemma_12}

\begin{proof}
	First, note that the derivative of the hypergeometric function with respect to $c$ is
	\[ %\begin{equation}\label{equation:F'}
	F'(a, b; c; z) =
	\sum_{n=0}^{\infty}
	\frac{(a)_n (b)_n}{(c)_n} 
	\frac{z^n}{n!}
	\big(\psi(c) - \psi(c + n)\big),
	\]
	where $\psi(x)$ denotes the digamma function ($\psi(x) = \frac{\mathrm{d}}{\mathrm{d}x} \ln \Gamma(x) = \frac{\Gamma'(x)}{\Gamma(x)}$).
	Since the digamma function increases on $(0, +\infty)$, we can see that $F'(a, b; c; z)$ is negative when parameters of the function are positive.
	
	Now, observe that since $\rho'(x) = \gamma - 1$, we have
	\[ %\begin{equation}
	F_1'(x) = (\gamma - 1) F'(1, 2; \rho(x); \gamma)
	\quad\textnormal{and}\quad
	F_2'(x) = (\gamma - 1) F'(2, 2; \rho(x); \gamma),
	\]
	and thus, given that $\gamma < 1$, they are both positive.
	In order to determine the sign of
	\[ %\begin{equation}\label{equation:R'}
	R'(x) = 
	\frac{F_2'(x) F_1(x) - F_2(x) F_1'(x)}
	{F_1(x)^2},
	\]
	we only need to determine the sign of its numerator.
	By considering the Cauchy product of $F_2'(x)$ and $F_1(x)$, we obtain
	\[ %\begin{equation}
	F_2'(x) F_1(x) =
	(\gamma - 1) \sum_{k=0}^{\infty} \gamma^k \sum_{n=0}^{k}
	(n + 1)
	\frac{(2)_n}
	{(\rho(x))_n}
	\frac{(2)_{k-n}}
	{(\rho(x))_{k-n}}
	\big(\psi(\rho(x)) - \psi(\rho(x) + n)\big).
	\]
	Similarly, for $F_2(x)$ and $F_1'(x)$, we have
	\[ %\begin{equation}
	F_2(x) F_1'(x) = 
	(\gamma - 1) \sum_{k=0}^{\infty} \gamma^k \sum_{n=0}^{k}
	(k - n + 1)
	\frac{(2)_n}
	{(\rho(x))_n}
	\frac{(2)_{k-n}}
	{(\rho(x))_{k-n}}
	\big(\psi(\rho(x)) - \psi(\rho(x) + n)\big).
	\]
	Finally, we express the difference between these two expressions as
	\[ %\begin{equation}\label{equation:R'-numerator}
	\begin{split}
		&F_2'(x) F_1(x) - F_2(x) F_1'(x) = \\
		&\qquad\qquad
		(\gamma - 1) \sum_{k=0}^{\infty} \gamma^k \sum_{n=0}^{k}
		(2n - k)
		\frac{(2)_n}
		{(\rho(x))_n}
		\frac{(2)_{k-n}}
		{(\rho(x))_{k-n}}
		\big(\psi(\rho(x)) - \psi(\rho(x) + n)\big).
	\end{split}
	\]
	We now check the sign of the inner sum.
	Observe that the sum of two elements with indices $n=i$ and $n=k - i$ is
	\[ %\begin{equation}
	(2i - k)
	\frac{(2)_i}
	{(\rho(\theta))_i}
	\frac{(2)_{k-i}}
	{(\rho(x))_{k-i}}
	\big(
	\psi(\rho(x) + k - i) - \psi(\rho(x) + i)
	\big).
	\]
	Since the digamma function increases on $(0, +\infty)$, the inner sum is negative, which, together with $\gamma<1$, implies $F_2'(x) F_1(x) - F_2(x) F_1'(x) > 0$. We conclude that $R'(x) > 0$.
\end{proof}

\begin{repeat_lemma_13}
	The function $R(x)$ is a contraction mapping on $[0, \hat\theta]$.
\end{repeat_lemma_13}

\begin{proof}
	Remind that a function $f : S \mapsto S$, defined on a metric space $(S, d)$, is called a contraction mapping, if there exists a constant $q \in [0, 1)$, such that for all $s_1, s_2 \in S$, we have $d(f(s_1), f(s_2)) \leq q d(s_1, s_2)$.
	%	\[
	%		d(f(s_1), f(s_2)) \leq q d(s_1, s_2).
	%	\]
	If $f(s)$ is a differentiable function, such that $\sup\,|f'(s)| < 1$, then $f(s)$ is a contraction mapping with $q = \sup\,|f'(s)|$.
	
	Using the form of $R(x)$ presented in Lemma \ref{lemma:r_formula}, we obtain
	\[
	R'(x) = 
	1 + \frac{1}{1 - \gamma} \frac{ \rho'(x)      F_1 (x) - 
		(\rho (x) - 1) F_1'(x)}	{F_1(x)^2}.
	\]
	$F_1(x)$ is positive and increases, and $\rho(x) > 1$ and decreases on $[0, \hat\theta]$, which implies that the right term of the expression is negative.
	Since $R(x)$ also increases on $[0, \hat\theta]$ (Lemma \ref{lemma:R-increases}), we have that $|R'(\theta)| \in [0, 1)$ for any $\theta \in [0, \hat\theta]$.
	Therefore, by the extreme value theorem, we know that $|R'(\theta)|$ achieves some maximum value $q \in (0, 1)$.
	Then, since $[0, \hat\theta]$ is a complete metric space and $R([0, \hat\theta]) \subseteq [0, \hat\theta]$ (using the formula from Lemma \ref{lemma:r_formula} it is easy to check that $R(0)>0$ and $R(\hat\theta) < \hat\theta$), we conclude that $R(x)$ is a contraction mapping on $[0, \hat\theta]$.
\end{proof}

\newpage
\section{Further experimental results}
Below we present the results for simulated ${H}^* = H(H_0, p_v = 0.3, p_e = 0.5, p_d = 0.2, Y_t)$, where $Y_t$ follows a truncated Poisson distribution with mean $\lambda = 4$, i.e., for $k \in \{1,2,3,\ldots\}$
\[
	\Pr[Y_t = k] = \frac{\lambda^k}{(e^{\lambda}-1) k!}.
\] 

Figure \ref{fig:theta_poisson} shows the evolution of the average degree of a vertex selected for deactivation in ${H}^*$ compared with the value of $\theta$ calculated using the fixed-point iteration method (Corollary~\ref{cor:theta_est}). It shows the convergence of the empirical average degree of a deactivated vertex to the estimated value of $\theta$ which again supports both, our Assumption (4) as well as the method for evaluating~$\theta$ (Figure \ref{fig:theta_iter_poisson} presents its visualization).

	\begin{figure}[!ht]%{.49\textwidth}
	\vspace{-20pt}
	\centering
	\includegraphics[width=.67\linewidth]{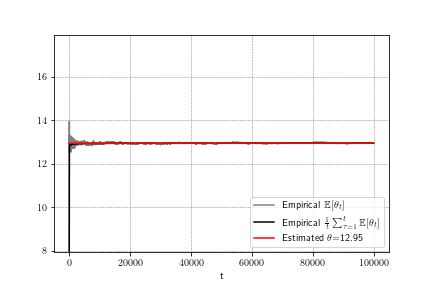}
	\caption{The empirical average deactivated degree (grey and black curves) as compared to the estimated $\theta$ (red line) in ${H}^*$.} 
	\label{fig:theta_poisson}
\end{figure}
%\hfill
\begin{figure}[!ht]%{.49\textwidth}
	\vspace{-40pt}
	\centering 
	\includegraphics[width=.67\linewidth]{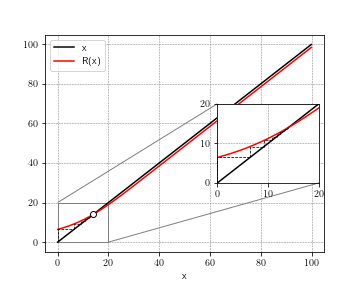}
	\caption{The visualization of the fixed-point iteration method applied to the model ${H}^*$, %$H(H_0, p_v = 0.3, p_e = 0.49, p_d = 0.21, Y_t = 3)$,
		starting from $\theta_0 = 0$.}
	\label{fig:theta_iter_poisson}
	%\vspace{-.4cm}
\end{figure}

Next, as for $\tilde{H}$ in Section \ref{sec:experiments}, we checked empirically the value of $\E[D_t]$ in ${H}^*$ (we ran again 1000 simulations up to 100000 steps). The empirical $\E[D_t]$ appeared to be linear, this time with the slope $\hat\alpha = 0.461354$ (Figure \ref{fig:Dt_EDt_poisson}). We then calculated the slope of the theoretical $\E[D_t]$ using the fixed-point iteration method to compute $\theta$ and then plugging it into equation from Lemma~\ref{fact:Dt/t}. It yielded $\alpha = 0.461397$ which closely corresponds to $\hat\alpha$.

\begin{figure}[!ht]%{.49\textwidth}
	%\vspace{-11pt}
	\centering
	\includegraphics[width=.67\linewidth]{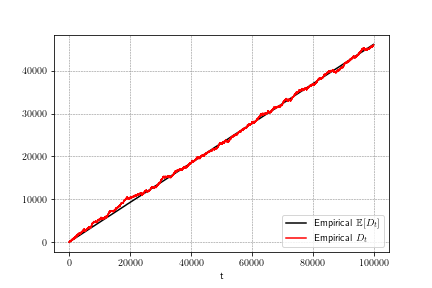}
	\caption{The empirical $\E[D_t]$ and the trajectory of $D_t$ for ${H}^*$.} 
	\label{fig:Dt_EDt_poisson}
\end{figure}
%\hfill
The result seen in Figure~\ref{fig:Dt_concentration_poisson} supports again Assumption~(3) about the concentration of $D_t$.

\begin{figure}[!ht]%{.49\textwidth}
	%\vspace{-11pt}
	\centering 
	\includegraphics[width=.67\linewidth]{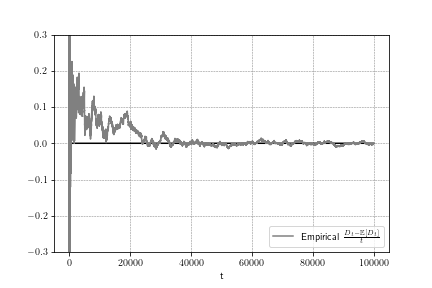}
	\caption{The concentration of $D_t$ for the~hypergraph ${H}^*$.}
	\label{fig:Dt_concentration_poisson}
	%\vspace{-.4cm}
\end{figure}

%%%%%%%%%%%%%%%%%%%%%%%%%%%%%%%%%%%%%%%%

\end{document}